\documentclass[11pt,a4paper]{article}

\usepackage[margin=25mm]{geometry}
\usepackage{amsmath,amssymb,amsthm}
\usepackage{mathtools}
\usepackage{booktabs,array}
\usepackage{enumitem}
\setlist[itemize]{label=$\bullet$}
\usepackage{cite}
\usepackage{float}
\usepackage{longtable}
\usepackage{tikz}
\usetikzlibrary{arrows.meta}
\usepackage{xurl}
\usepackage[hidelinks]{hyperref}
\usepackage[capitalise,nameinlink]{cleveref}

\newtheorem{theorem}{Theorem}[section]
\newtheorem{proposition}[theorem]{Proposition}
\newtheorem{lemma}[theorem]{Lemma}
\newtheorem{corollary}[theorem]{Corollary}
\theoremstyle{definition}
\newtheorem{definition}[theorem]{Definition}
\newtheorem{example}[theorem]{Example}
\newtheorem{assumption}[theorem]{Assumption}
\theoremstyle{remark}

\newcommand{\FOUR}{\mathbf{4}}
\newcommand{\Yv}{\mathsf{Y}}
\newcommand{\Nv}{\mathsf{N}}
\newcommand{\Uv}{\mathsf{U}}
\newcommand{\Bv}{\mathsf{YN}}
\newcommand{\lek}{\le_{k}}
\newcommand{\leqt}{\le_{t}}
\newcommand{\req}{\mathrm{req}}
\newcommand{\Stab}{\mathrm{Stab}}
\newcommand{\Phiw}{\Phi^{\simeq}}
\newcommand{\cX}{\mathcal{X}}
\newcommand{\Act}{\mathrm{Act}}
\newcommand{\Sanc}{\mathrm{Sanc}}
\title{Epistemic Disturbance in the Graph Model for Conflict Resolution:\\
State-Preserving Actions, Four-Valued Assessments,\\
and the Distinction between Capability and Intention}

\author{Yukiko Kato\\[2pt]
\small Lynx Technologies Inc., Tokyo, Japan\\
\small Institute of Science Tokyo, Tokyo, Japan\\
\small \texttt{ykato@lynxjp.com}}

\date{Draft}

\begin{document}
\maketitle

\begin{abstract}
In the graph model for conflict resolution (GMCR), a decision maker (DM)
either moves the conflict to another state or does nothing. The basic
definitions leave inaction implicit: remaining at a state is not a move and
receives no analysis of its own, so every action that leaves the state
unchanged is treated as doing nothing. Many actions in real conflicts are
neither moves nor inaction: announcements, exercises, leaks and selective
disclosures leave the state of the conflict unchanged but change what other
DMs believe about the moves available to whom and about the moves each DM
would want to make. We introduce such state-preserving actions into GMCR by
augmenting states with the DMs' epistemic states. A physical move changes
the physical state; a state-preserving action changes only the epistemic
state; inaction is the absence of a transition. A state-preserving action
generates evidence in each observer through an observer-specific
interpretation map, so that the same action may be read differently by
different DMs. To represent assessments that may lack evidence or contain
conflicting evidence, and in which evidence accumulates, we build on a
four-valued extension of GMCR proposed in the author's earlier work. We show
that evidence for the availability of a move can only enable perceived
moves and evidence against can only disable them, that two of the four
reduction operators ignore one kind of evidence entirely, and that
contradictory assessments are absorbing under monotone accumulation of
evidence. Combined with the monotonicity of stability in move sets, these
facts determine the direction in which any state-preserving action moves a
DM's stability judgements, and characterize the channels through which an
action can enable provocation or deterrence. Capability assessments affect
all sanction-based stability concepts and, on the focal DM's own side, Nash
stability as well, whereas intention assessments affect only sequential
stability. Hedging between two candidate
types weakly expands the set of states an observer judges sequentially
stable if the observer reads contradictory evidence as credibility, and
weakly shrinks it otherwise. The 1995 negotiation that unified the two
competing DVD formats illustrates the distinction: general metarationality
cannot distinguish the phases of the negotiation, because the technical
working group of computer companies could always sanction, whereas
sequential stability, which asks whether it would, can.
\end{abstract}

% ===============================================================
\section{Introduction}
% ===============================================================

A graph model records the state of a conflict as a combination of options
taken by the decision makers (DMs), and a move as a change of that
combination \cite{KHF1987,FHK1993}. In the basic definitions, a move is an
arc between two distinct states, and the reachable lists from which
stability is analysed exclude the current state; excluding it is a standard
convention in the literature \cite{RegoVieira2022}. Remaining at a state is
therefore not a move. It is not examined in its own right either: it enters
the definitions only implicitly, as what a DM does when it does not move,
and stability analysis asks whether a DM has reason to leave a state, never
what a DM does while it stays. Everything that happens without a change of
state is thereby treated as one and the same: doing nothing.

This implicit treatment of inaction is the starting point of this paper.
Many actions in real conflicts leave the combination of options unchanged
and still matter: a declaration that a sanction is ready, a
military exercise, a leak suggesting that a threat is empty, a private
signal to one party but not another, a deceptive signal about future capability
or intention. What such an action changes is
not the state of the conflict but what other DMs believe about which moves
are available and which moves the others would want to make.

This paper introduces such actions into the graph model. The central
distinction is between three things that the physical state alone cannot
tell apart. Let $s$ be a physical state and $\kappa$ a profile of the DMs'
epistemic states. Then
\[
\begin{array}{lll}
(s,\kappa)\longrightarrow(t,\kappa) & s\neq t & \text{physical move},\\[2pt]
(s,\kappa)\longrightarrow(s,\kappa') & \kappa\neq\kappa' & \text{state-preserving action},\\[2pt]
(s,\kappa) & \text{no transition} & \text{inaction}.
\end{array}
\]
Under the projection onto the physical state, a state-preserving action
looks like a self-loop $s\to s$, and in the basic definitions of the graph
model it is indistinguishable from inaction. It is not a null action; it is
null only under that projection. Making inaction explicit, as the absence of
a transition, is what allows state-preserving actions to be defined at all.

To represent the phenomena considered here, we distinguish at least three
epistemically different situations that a binary representation (available
or not) collapses: the absence of any evidence, evidence in one direction,
and conflicting evidence. We also let evidence be added without erasing what
was there before, since an announcement or a leak adds to what others know
rather than replacing it. A four-valued extension of GMCR proposed by the
author in earlier work \cite{Kato2022,KatoQCW} is a particularly economical
way to make these distinctions, because it separates evidence for from
evidence against explicitly. It uses Belnap's
four values \cite{Belnap1977,Dunn1976}: evidence only for ($\Yv$), only
against ($\Nv$), neither ($\Uv$) or both ($\Bv$). This extension is recent
and not yet widely adopted, and we do not presuppose it as the standard
representation of beliefs in the graph model. We adopt it because it makes
the situation studied here straightforward to model and its effects
explicit. Probabilities could also express uncertainty, but a single
precise probability does not in general distinguish the absence of evidence
from balanced conflicting evidence without additional structure, whereas
$\Uv=(0,0)$ and $\Bv=(1,1)$ are distinct values here; precise probabilities
also require quantitative belief information that the standard graph model
does not require.
Richer formalisms such as imprecise probabilities or belief functions can
make the distinction, at the cost of further quantitative structure. Hypergame models
represent misperception through subjective games and generally do not
separate the lack of evidence from conflicting evidence. When every assessment is $\Yv$ or
$\Nv$, the framework reduces to the standard graph model with perceived move
sets.

A state-preserving action generates evidence in an observer, and the same
action may generate different evidence in different observers. We therefore
separate the action from its effect: each observer has an
\emph{interpretation map} that turns an action into a change of its own
assessment. A DM acts on its four-valued assessments through a
\emph{reduction operator} that sends $\Uv$ and $\Bv$ to $\Yv$ or $\Nv$; the
four operators of this kind react to evidence in structurally different
ways. Our results are as follows.
\begin{itemize}[leftmargin=*]
  \item We make explicit the treatment of inaction, which the basic
        definitions of the graph model leave implicit, and distinguish it
        from state-preserving actions by augmenting states with epistemic
        states (\cref{sec:actions}). Inaction is the absence of a transition;
        a state-preserving action is a transition that changes only the
        epistemic state; the convention that remaining at a state is not a
        move is preserved.
  \item Evidence for the availability of a move can only enable perceived
        moves, and evidence against can only disable them, for all four
        operators (\cref{prop:truth}); two operators ignore one of the two
        kinds of evidence entirely (\cref{prop:sensitivity}); and the
        assessments that no further evidence can change, as far as action
        is concerned, are the contradictory ones and the definite ones that
        the operator resolves in the same way as contradiction
        (\cref{prop:invariant}).
  \item Stability is monotone in the opponents' move sets and antitone in
        the focal DM's own (\cref{lem:stabmono}). Together these facts
        determine the direction in which a state-preserving action moves a
        DM's judgement of its own stability (\cref{thm:direction,thm:credibility}),
        and give a sufficient condition for a judgement to be immune to
        every action (\cref{prop:certificate}).
  \item We define actions that enable provocation or deterrence, and show
        which operators are open to which, through which channel
        (\cref{cor:channels}).
  \item Capability assessments affect all sanction-based stability
        concepts and, on the focal DM's own side, Nash stability as well,
        whereas intention assessments affect only sequential stability
        (\cref{thm:direction,thm:credibility}, \cref{tab:dimensions}).
  \item When intention assessments derive from evidence about candidate
        types of a DM, hedging between two types makes exactly the disputed
        moves contested, and weakly expands or weakly shrinks an observer's
        set of sequentially stable states according to how the observer
        resolves contradiction (\cref{thm:hedging}).
  \item In the 1995 DVD format negotiation, general metarationality cannot
        distinguish the phases, because the working group of computer
        companies could always sanction, whereas sequential stability can,
        because it asks whether the group would (\cref{sec:dvd}).
\end{itemize}

The central result, \cref{thm:direction}, can be summarized as follows,
where ``stabilizing'' means that the set of states a DM judges stable can
only grow:
\[
\begin{array}{l|cc}
 & \text{evidence for} & \text{evidence against}\\ \hline
\text{about opponents' moves} & \text{stabilizing} & \text{destabilizing}\\
\text{about own moves} & \text{destabilizing} & \text{stabilizing}
\end{array}
\]

How DMs anticipate the evolution of the conflict under their assessments,
and how an actor can plan a sequence of state-preserving actions, is the
subject of a companion paper \cite{KatoAnticipation}; the present paper
concerns the effect of state-preserving actions on stability judgements.

% ===============================================================
\section{Related work and positioning}\label{sec:related}
% ===============================================================

\subsection{Moves and inaction in the graph model}

The graph model represents a conflict by states, the moves that each DM
controls between them, and the DMs' preferences \cite{KHF1987,FHK1993}. Moves
are arcs between distinct states, and the reachable lists used in stability
analysis exclude the current state, a convention stated explicitly in later
work \cite{RegoVieira2022}. Stability concepts ask whether a DM has an
improvement that it would take, and status quo analysis traces the
evolution of a conflict through sequences of moves \cite{Xu2010}. In both,
what a DM does without changing the state is not represented: it coincides
with doing nothing. To our knowledge, the distinction between inaction and
actions that leave the state unchanged has not been drawn in the graph
model. The categorical formulation of \cite{KatoCGMCR} makes identities
explicit as morphisms, which suggests the distinction, but does not model
what such actions change.

\subsection{Misperception and unawareness}

Misperception of options and preferences in the graph model has been
studied through hypergames \cite{WHF1989,Aljefri2018}, generalized to
several levels of misperception \cite{Aljefri2020}, through perceptual graph
models \cite{Obeidi2009}, and through interactive unawareness, including
higher-order awareness \cite{RegoVieira2022}. These frameworks take a
structure of misperception or unawareness as given and ask what stability
and equilibria follow from it; in the model of \cite{RegoVieira2022}, awareness
may also increase along moves. None of them treats actions whose purpose and
effect is to change the others' perceptions while the state stays the same.

\subsection{Manipulation and deception}

Within the graph model, strategic manipulation has been studied as the
costly adjustment of other parties' preferences through a social trust
network, formulated as an optimization problem \cite{Zhang2025}, and
Inohara et al.\ showed, in a framework of interdependent preferences,
conditions under which deceit about preferences cannot succeed
\cite{Inohara1997}. Outside the graph model, deception in $2\times2$ games
\cite{Brams1977}, deception in hypergames, including stealthy deception in
which a player manipulates another's perception without being detected
\cite{Gharesifard2014,Kovach2015}, and strategic information transmission
\cite{CrawfordSobel1982} concern related phenomena, typically with
cardinal payoffs or probabilistic beliefs.

\subsection{Representations of evidence}

Belnap's four-valued logic was designed for reasoning under possibly
inconsistent incoming information \cite{Belnap1977,Dunn1976}, and its
bilattice structure separates the order of truth from the order of
information \cite{Ginsberg1988}. The four-valued extension of the graph
model proposed in \cite{Kato2022,KatoQCW} uses it to represent the status of
options, with reduction operators that turn four-valued assessments into
decisions. It has so far been used for given assessments, not for their
change through actions.

\subsection{Positioning}

Against this background, the contributions of this paper are the
following. First, it makes explicit the treatment of inaction that the
basic definitions leave implicit, and distinguishes it from
state-preserving actions within the graph model itself. Second, where the
literature on misperception analyses given structures of perception, this
paper analyses how state-preserving actions transform them, through
observer-specific interpretation maps that cover both the common signal read
differently and selective disclosure. Third, where manipulation in the
graph model has concerned preferences, this paper keeps preferences fixed
and acts on assessments of capability and intention, and it shows that
these two kinds of assessment act on different stability concepts. Fourth,
the direction of each effect is determined ordinally, by the kind of
evidence and the observer's reduction operator, without cardinal payoffs or
probabilities; the results depend on the four-valued extension only through
the evidence-pair structure of \cref{sec:operators}.

% ===============================================================
\section{Preliminaries}
% ===============================================================

\subsection{Graph models, sanctions and stability}

We follow the conventions of \cite{KatoExclusion}. A graph model
$G=(N,S,(A_i)_{i\in N},(\succsim_i)_{i\in N})$ has a finite set $N$ of DMs,
a finite set $S$ of states, arc sets $A_i\subseteq S\times S$ and complete,
not necessarily transitive, preference relations $\succsim_i$ on $S$, with
strict part $\succ_i$. Remaining at a state is always possible and is not a
move. For $s\in S$,
\[
  R_i(s)=\{t : (s,t)\in A_i,\ t\neq s\},\qquad
  R_i^{+}(s)=\{t\in R_i(s): t\succ_i s\},\qquad
  \Phiw_i(s)=\{t: s\succsim_i t\}.
\]
For $i\in N$, $R_{N-i}(s)$ and $R_{N-i}^{+}(s)$ are the states reachable
from $s$ by legal sequences of moves, respectively improvements, of the DMs
other than $i$, in which no DM moves twice in succession \cite{FHK1993}.

We use the four basic stability concepts of the graph model
\cite{KHF1987,FHK1993} and state them through an explicit sanction
predicate. Let $i\in N$, $s\in S$ and $t\in R_i^{+}(s)$. The improvement
from $s$ to $t$ has
\begin{itemize}[leftmargin=*]
  \item a \emph{GMR-sanction} if some $s''\in R_{N-i}(t)$ satisfies
        $s\succsim_i s''$;
  \item an \emph{SMR-sanction} if some such $s''$ also satisfies
        $R_i(s'')\subseteq\Phiw_i(s)$;
  \item an \emph{SEQ-sanction} if some $s''\in R^{+}_{N-i}(t)$ satisfies
        $s\succsim_i s''$.
\end{itemize}
We write $\Sanc^\sigma_i(s,t)$ when the improvement has a $\sigma$-sanction.
State $s$ is \emph{Nash stable} for $i$ \cite{Nash1950,Nash1951} if
$R_i^{+}(s)=\emptyset$, and \emph{general metarational} (GMR)
\cite{Howard1971}, \emph{symmetric metarational} (SMR) \cite{Howard1971} or
\emph{sequentially stable} (SEQ) \cite{FraserHipel1979,FraserHipel1984} for
$i$ if every $t\in R_i^{+}(s)$ satisfies $\Sanc^\sigma_i(s,t)$ for
$\sigma=\mathrm{GMR}$, $\mathrm{SMR}$ or $\mathrm{SEQ}$ respectively. Let
$\cX=\{\mathrm{Nash},\mathrm{GMR},\mathrm{SMR},\mathrm{SEQ}\}$ and write
$\Stab^\sigma_i(G)$ for the set of states $\sigma$-stable for $i$ in $G$.
Coalitional versions of these concepts \cite{Kilgour2001,Inohara2008} are not
treated in this paper.

For $s\in S$ let $\hat D_G(s)$ be the forward-reachable set of $s$, the
smallest set containing $s$ and closed under arc successors. By the
localization lemma \cite[Lemma~3.4, Cor.~3.5]{KatoExclusion}, the
$\sigma$-stability of $s$ for any DM is determined by $\hat D_G(s)$, the arcs
with tails in $\hat D_G(s)$ and the preferences restricted to
$\hat D_G(s)$.

\subsection{Four values as evidence pairs}

We write each element of Belnap's lattice $\FOUR$ as a pair $(f,a)\in\{0,1\}^2$
recording whether there is evidence for ($f$) and against ($a$):
\[
  \Yv=(1,0),\qquad \Nv=(0,1),\qquad \Uv=(0,0),\qquad \Bv=(1,1).
\]
The \emph{knowledge order} is componentwise,
$(f,a)\lek(f',a')$ iff $f\le f'$ and $a\le a'$, with bottom $\Uv$ and top
$\Bv$. The \emph{truth order} is $(f,a)\leqt(f',a')$ iff $f\le f'$ and
$a\ge a'$, with bottom $\Nv$ and top $\Yv$. Conjunction is the truth-order
meet $(f,a)\wedge(f',a')=(\min(f,f'),\max(a,a'))$, and the empty conjunction
is $\Yv$ \cite{Belnap1977,Ginsberg1988}.

\begin{lemma}\label{lem:conj}
Conjunction is monotone in both the knowledge order and the truth order.
The $f$-component of a conjunction depends only on the $f$-components of its
arguments, and the $a$-component only on their $a$-components.
\end{lemma}
\begin{proof}
Immediate from the componentwise formula.
\end{proof}

\subsection{Reduction operators}

A \emph{reduction operator} is a map $\delta\colon\FOUR\to\{\Yv,\Nv\}$ with
$\delta(\Yv)=\Yv$ and $\delta(\Nv)=\Nv$. There are exactly four, determined by
$(\delta(\Uv),\delta(\Bv))$; we write $\delta_{\Nv\Yv}$ for the one with
$\delta(\Uv)=\Nv$ and $\delta(\Bv)=\Yv$, and similarly $\delta_{\Yv\Nv}$,
$\delta_{\Nv\Nv}$ and $\delta_{\Yv\Yv}$. In terms of evidence,
\begin{align*}
  \delta_{\Nv\Yv}(f,a)=\Yv &\iff f=1
    &&\text{(any evidence for suffices)},\\
  \delta_{\Yv\Nv}(f,a)=\Yv &\iff a=0
    &&\text{(any evidence against defeats)},\\
  \delta_{\Nv\Nv}(f,a)=\Yv &\iff f=1\ \text{and}\ a=0
    &&\text{(only uncontested evidence for)},\\
  \delta_{\Yv\Yv}(f,a)=\Yv &\iff f=1\ \text{or}\ a=0
    &&\text{(anything but uncontested evidence against)}.
\end{align*}
The operator $\delta_{\Nv\Yv}$ is the operator $\delta_{\mathrm{prec}}$ of
\cite{KatoQCW}, the only one that is monotone in the knowledge order.
All four are monotone in the truth order \cite{KatoQCW}. The notation
differs from that of \cite{KatoQCW}: here each operator is indexed directly
by its values on $\Uv$ and $\Bv$.

% ===============================================================
\section{State-preserving actions}\label{sec:actions}
% ===============================================================

\subsection{Epistemic states and perceived models}

Let $O=\bigsqcup_{i\in N}O_i$ be a finite set of options, $O_i$ being those
controlled by DM $i$. Each arc $a\in A_i$ carries a \emph{requirement}
$\req(a)\subseteq O_i$, the options of $i$ whose availability the move
depends on. For $i\in N$ let $A_{-i}=\bigcup_{j\neq i}A_j$.

\begin{definition}[Epistemic state]\label{def:epistemic}
The \emph{epistemic state} of DM $i$ is a pair $\kappa^i=(\psi^i,\chi^i)$ of
\begin{itemize}[leftmargin=*]
  \item a \emph{capability assessment} $\psi^i\colon O\to\FOUR$, recording
        the evidence $i$ holds on the availability of each option; and
  \item an \emph{intention assessment} $\chi^i\colon A_{-i}\to\FOUR$,
        recording, for each move of another DM, the evidence $i$ holds that
        the move would be an improvement for the DM who controls it.
\end{itemize}
The set of epistemic states of DM $i$ is
$K_i=\FOUR^{O}\times\FOUR^{A_{-i}}$, and an \emph{epistemic profile} is an
element $\kappa=(\kappa^i)_{i\in N}$ of the finite set
$K=\prod_{i\in N}K_i$. We write $\kappa^i\lek\kappa'^i$ if the knowledge order
holds in every component of both assessments.
\end{definition}

Each DM $i$ uses a reduction operator $\delta_i$. The \emph{warrant} of an arc
$a$ under $\psi$ is $w_\psi(a)=\bigwedge_{o\in\req(a)}\psi(o)$.

\begin{definition}[Perceived model and judgements]\label{def:perceived}
The \emph{perceived arcs} of DM $i$ are
$A_j[\kappa^i]=\{a\in A_j:\delta_i(w_{\psi^i}(a))=\Yv\}$ for $j\in N$, and its
\emph{perceived credible arcs} are
$B_j[\kappa^i]=\{a\in A_j[\kappa^i]:\delta_i(\chi^i(a))=\Yv\}$ for $j\neq i$,
while $B_i[\kappa^i]$ consists of the arcs of $A_i[\kappa^i]$ that are
improvements for $i$. The \emph{perceived model} $G[\kappa^i]$ is the graph
model with arcs $A_j[\kappa^i]$. DM $i$ \emph{judges} $s$ to be
$\sigma$-stable for itself if $s$ is $\sigma$-stable for $i$ in $G[\kappa^i]$,
where for $\sigma=\mathrm{SEQ}$ the sets $R^{+}_{N-i}$ are replaced by the
sets $R^{B}_{N-i}$ of states reachable by legal sequences along perceived
credible arcs. We write $\Stab^\sigma_i[\kappa^i]$ for the set of states $i$
judges $\sigma$-stable.
\end{definition}

Preferences are taken to be those of the model; what a DM may misjudge is
which moves are available and which moves the others would want to make.
When $\chi^i(a)=\Yv$ exactly for the arcs that are improvements for their
controllers and $\chi^i(a)=\Nv$ otherwise, $B_j[\kappa^i]$ is the set of
perceived improving arcs, and the judgements are standard stability in the
perceived model. When, moreover, every capability assessment is $\Yv$ or
$\Nv$ and agrees with the model, the perceived model is the model itself.

\subsection{Three kinds of transition}

\begin{definition}[Augmented state space and transitions]\label{def:augmented}
The \emph{augmented state space} is $\widetilde S=S\times K$, with projection
$\pi_S(s,\kappa)=s$. A transition of $\widetilde S$ is one of the following.
\begin{enumerate}[label=(\alph*)]
  \item A \emph{physical move} of DM $i$: $(s,\kappa)\to(t,\kappa)$ with
        $(s,t)\in A_i$ and $t\neq s$.
  \item A \emph{state-preserving action} of DM $k$:
        $(s,\kappa)\to(s,\kappa')$ with $\kappa'\neq\kappa$, as specified in
        \cref{def:action}.
\end{enumerate}
\emph{Inaction} at $(s,\kappa)$ is the absence of a transition.
\end{definition}

In the present model physical moves leave epistemic states unchanged: we
abstract from the epistemic updating that observed physical moves cause in
practice, which would require an additional observation rule. Under
$\pi_S$, a state-preserving action is sent to a self-loop at $s$, while
inaction involves no transition at all. The convention that remaining at a
state is not a move, and hence that $R_i(s)$ excludes $s$, concerns
inaction; it does not exclude state-preserving actions, which are
transitions of the augmented space and do not enter the legal sequences of
the physical model.

\begin{definition}[Actions and interpretation]\label{def:action}
For each DM $k$ and state $s$, let $\Act_k(s)$ be a finite set of
\emph{state-preserving actions} available to $k$ at $s$. For each observer
$i\neq k$, an \emph{interpretation map}
$\Gamma_i\colon\Act_k(s)\times K_i\to K_i$ assigns to each $d\in\Act_k(s)$ and
epistemic state $\kappa^i$ an epistemic state $\Gamma_i(d,\kappa^i)$ with
\[
  \kappa^i\lek\Gamma_i(d,\kappa^i).
\]
Performing $d$ at $(s,\kappa)$ leads to $(s,\kappa')$, where
$\kappa'^i=\Gamma_i(d,\kappa^i)$ for $i\neq k$ and $\kappa'^k=\kappa^k$. The
change $\kappa^i\mapsto\Gamma_i(d,\kappa^i)$ is the \emph{disturbance} of
$i$'s epistemic state \emph{induced} by $d$, and
$\mathcal D_{k\to i}(s,\kappa^i)=\{\Gamma_i(d,\kappa^i):d\in\Act_k(s)\}$ is
the set of disturbances of $i$ that $k$ can induce at $s$. The dependence of
$\Gamma_i$ on the actor $k$ and the physical state $s$ is suppressed in the
notation.
\end{definition}

We assume that a state-preserving action does not directly alter the
actor's own epistemic state. In practice an actor may learn from the
reactions to its own action, as with a probe or a trial announcement;
allowing such self-updating is a straightforward extension that we do not
pursue here.

An action is not itself evidence; it generates evidence in each observer,
and different observers may read the same action differently. If $i$ does
not observe $d$, then $\Gamma_i(d,\kappa^i)=\kappa^i$; selective disclosure
is an action observed by some DMs and not others. Announcements, exercises,
leaks and bluffs are distinguished by their sets $\Act_k(s)$ and their
interpretation maps, not by their effect on the physical state, which is
none. The requirement $\kappa^i\lek\Gamma_i(d,\kappa^i)$ is the following
assumption.

\begin{assumption}[Monotone accumulation]\label{ass:accumulation}
Evidence can be added but not removed: interpretation never lowers an
assessment in the knowledge order.
\end{assumption}

The assumption excludes retraction, correction and forgetting of evidence,
and the discrediting of sources; its consequences below, in particular the
absorbing character of contradiction, depend on it.

\subsection{Kinds of disturbance}

A disturbance $\kappa^i\mapsto\kappa'^i$ changes the capability assessment
$\psi^i$, the intention assessment $\chi^i$, or both. A change of $\psi^i$
is \emph{opponent-side} if it concerns only options outside $O_i$ and
\emph{own-side} if it concerns only options in $O_i$. A change of either
assessment is a \emph{for-disturbance} if it raises only $f$-components and
an \emph{against-disturbance} if it raises only $a$-components. Since the
components are raised independently, every disturbance is a composite, in
any order, of a for-part and an against-part, each of which splits further
into opponent-side and own-side parts and into capability and intention
parts. These componentwise updates commute at the level of epistemic states
and lead to the same final assessment; the intermediate perceived models
may differ with the order (for instance, $\Uv\to\Yv\to\Bv$ and
$\Uv\to\Nv\to\Bv$).

% ===============================================================
\section{How reduction operators respond to evidence}\label{sec:operators}
% ===============================================================

\begin{proposition}[Direction of evidence]\label{prop:truth}
For every reduction operator $\delta$ and every $x\in\FOUR$: raising the
$f$-component of $x$ can change $\delta(x)$ only from $\Nv$ to $\Yv$, and
raising the $a$-component can change it only from $\Yv$ to $\Nv$.
\end{proposition}
\begin{proof}
Raising $f$ is an increase in the truth order and raising $a$ a decrease,
and every reduction operator is monotone in the truth order. Directly: in
each of the four evidence formulas, the condition for $\Yv$ is preserved by
raising $f$ and by lowering $a$.
\end{proof}

\begin{proposition}[Sensitivity]\label{prop:sensitivity}
$\delta_{\Nv\Yv}(x)$ depends only on the $f$-component of $x$, and
$\delta_{\Yv\Nv}(x)$ only on the $a$-component. Each of $\delta_{\Nv\Nv}$ and
$\delta_{\Yv\Yv}$ depends on both components.
\end{proposition}
\begin{proof}
The first two claims are read off the evidence formulas. For
$\delta_{\Nv\Nv}$, raising $f$ at $\Uv$ changes the value from $\Nv$ to $\Yv$,
and raising $a$ at $\Yv$ changes it from $\Yv$ to $\Nv$. For $\delta_{\Yv\Yv}$,
raising $f$ at $\Nv$ changes it from $\Nv$ to $\Yv$, and raising $a$ at $\Uv$
from $\Yv$ to $\Nv$.
\end{proof}

Since a disturbance raises both components in general, these facts recover
the behaviour of the operators in the knowledge order: $\delta_{\Nv\Yv}$ is
monotone, $\delta_{\Yv\Nv}$ is antitone, and $\delta_{\Nv\Nv}$ and
$\delta_{\Yv\Yv}$ are neither. This is the pattern behind the failure of
knowledge monotonicity for three of the four operators in \cite{KatoQCW}.

\begin{proposition}[Immune assessments]\label{prop:invariant}
Call $x\in\FOUR$ \emph{immune} for $\delta$ if $\delta(y)=\delta(x)$ for every
$y$ with $x\lek y$. Then $\Uv$ is immune for no operator, $\Bv$ is immune for
every operator, $\Yv$ is immune iff $\delta(\Bv)=\Yv$, and $\Nv$ is immune iff
$\delta(\Bv)=\Nv$. Hence the immune values are $\{x:f=1\}$ for $\delta_{\Nv\Yv}$
and $\delta_{\Yv\Yv}$, and $\{x:a=1\}$ for $\delta_{\Yv\Nv}$ and
$\delta_{\Nv\Nv}$.
\end{proposition}
\begin{proof}
The upper sets are $\FOUR$ for $\Uv$, $\{\Bv\}$ for $\Bv$, $\{\Yv,\Bv\}$ for
$\Yv$ and $\{\Nv,\Bv\}$ for $\Nv$, and $\delta$ is not constant on $\FOUR$.
\end{proof}

Under \cref{ass:accumulation}, ignorance is always open to further
evidence, and a contradictory assessment is \emph{absorbing}: it cannot be
changed by any further evidence. This is a consequence of monotone
accumulation, not of four-valued logic alone. A definite assessment resists
further evidence exactly when the DM resolves contradiction in the same
direction: a DM that treats contested evidence as availability cannot be
talked out of believing that a move is available, but can be talked into
believing it; the dual DM the other way round.

\begin{lemma}[Immune warrants]\label{lem:immunewarrant}
Let $a$ be an arc and $\psi\lek\psi'$. If $\psi(o)$ is immune for $\delta$ for
every $o\in\req(a)$, then $\delta(w_{\psi'}(a))=\delta(w_\psi(a))$.
\end{lemma}
\begin{proof}
If $\req(a)=\emptyset$ both warrants are $\Yv$. Otherwise, if $\delta(\Bv)=\Yv$,
every $\psi(o)$ and every $\psi'(o)$, $o\in\req(a)$, has $f=1$ by
\cref{prop:invariant}; by \cref{lem:conj} both warrants have $f=1$, and
$\delta_{\Nv\Yv}$ and $\delta_{\Yv\Yv}$ send every such value to $\Yv$. If
$\delta(\Bv)=\Nv$, the same argument applies with $a=1$.
\end{proof}

% ===============================================================
\section{Stability and move sets}
% ===============================================================

\begin{lemma}[Monotonicity of reachability]\label{lem:reachmono}
Let $G$ and $G'$ have the same states and preferences, let $i\in N$, and
suppose that $A_j\subseteq A'_j$ for every $j\neq i$. Then
$R_{N-i}(s)\subseteq R'_{N-i}(s)$ and $R^{+}_{N-i}(s)\subseteq R'^{+}_{N-i}(s)$
for every $s$. The same holds for reachability along any family of
designated arc sets with $B_j\subseteq B'_j$ for every $j\neq i$.
\end{lemma}
\begin{proof}
A legal sequence in $G$ is a legal sequence in $G'$: its steps are arcs of
$G'$, and its labels and preference comparisons are unchanged.
\end{proof}

\begin{lemma}[Stability in move sets]\label{lem:stabmono}
Let $G$ and $G'$ have the same states and preferences, and let $i\in N$ and
$\sigma\in\cX$.
\begin{enumerate}[label=(\alph*)]
  \item If $A_i=A'_i$ and $A_j\subseteq A'_j$ for all $j\neq i$, then
        $\Stab^\sigma_i(G)\subseteq\Stab^\sigma_i(G')$, with equality for
        Nash.
  \item If $A_j=A'_j$ for all $j\neq i$ and $A_i\subseteq A'_i$, then
        $\Stab^\sigma_i(G')\subseteq\Stab^\sigma_i(G)$.
\end{enumerate}
The same holds for SEQ judgements with the credible arc sets $B_j$ in place
of the improving arcs of $j\neq i$.
\end{lemma}
\begin{proof}
(a) The improvements $R_i^{+}(s)$ and the escape sets $R_i(s'')$ depend only
on $A_i$; by \cref{lem:reachmono} the sets in which sanctions are sought can
only grow, so every sanction persists. Nash stability involves $A_i$ only.
(b) Now $R_i^{+}(s)$ and every $R_i(s'')$ can only grow, while the sets in
which sanctions are sought are unchanged; every condition requires a
sanction for every improvement, and SMR also requires
$R_i(s'')\subseteq\Phiw_i(s)$, so both become harder to satisfy.
\end{proof}

% ===============================================================
\section{The direction of disturbance}\label{sec:direction}
% ===============================================================

\begin{theorem}[Direction for capability]\label{thm:direction}
Let DM $i$ use $\delta_i$, let $\kappa^i\mapsto\kappa'^i$ be a disturbance that
changes only the capability assessment, let $\sigma\in\cX$, and write
$\Stab=\Stab^\sigma_i[\kappa^i]$ and $\Stab'=\Stab^\sigma_i[\kappa'^i]$.
\begin{enumerate}[label=(\alph*)]
  \item Opponent-side for-disturbance: $\Stab\subseteq\Stab'$.
  \item Opponent-side against-disturbance: $\Stab'\subseteq\Stab$.
  \item Own-side for-disturbance: $\Stab'\subseteq\Stab$.
  \item Own-side against-disturbance: $\Stab\subseteq\Stab'$.
\end{enumerate}
For Nash stability, (a) and (b) hold with equality.
\end{theorem}
\begin{proof}
Under a for-disturbance each changed assessment has its $f$-component raised
and its $a$-component unchanged; by \cref{lem:conj} the same holds for every
warrant, and by \cref{prop:truth} every perceived arc remains perceived.
Hence $A_j[\kappa^i]\subseteq A_j[\kappa'^i]$ for every $j$, with equality
for the DMs whose options are unaffected. Since the intention assessment is
unchanged, it follows that $B_j[\kappa^i]\subseteq B_j[\kappa'^i]$ for every
$j$ as well.
For an against-disturbance the inclusions are reversed. The four cases
follow from \cref{lem:stabmono}.
\end{proof}

\begin{corollary}[Operator-specific direction]\label{cor:operator}
Let $\kappa^i\mapsto\kappa'^i$ be an arbitrary opponent-side disturbance of
the capability assessment. If $i$ uses $\delta_{\Nv\Yv}$, then
$\Stab\subseteq\Stab'$; if $i$ uses $\delta_{\Yv\Nv}$, then
$\Stab'\subseteq\Stab$. For own-side disturbances the inclusions are
reversed. For $\delta_{\Nv\Nv}$ and $\delta_{\Yv\Yv}$ no such conclusion holds
in general.
\end{corollary}
\begin{proof}
Decompose the disturbance into a for-part and an against-part. Under
$\delta_{\Nv\Yv}$ perceived arcs depend only on $f$-components
(\cref{lem:conj,prop:sensitivity}), so the against-part has no effect; the
case $\delta_{\Yv\Nv}$ is symmetric. \Cref{ex:threat} shows both directions
for the other two operators.
\end{proof}

\begin{theorem}[Direction for intention]\label{thm:credibility}
Let $\kappa^i\mapsto\kappa'^i$ be a disturbance that changes only the
intention assessment.
\begin{enumerate}[label=(\alph*)]
  \item $i$'s Nash, GMR and SMR judgements do not change.
  \item Under a for-disturbance, $\Stab^{\mathrm{SEQ}}_i[\kappa^i]\subseteq
        \Stab^{\mathrm{SEQ}}_i[\kappa'^i]$; under an against-disturbance the
        reverse inclusion holds.
  \item Under $\delta_{\Nv\Yv}$ every such disturbance can only enlarge the set
        of states judged SEQ; under $\delta_{\Yv\Nv}$ it can only shrink it.
\end{enumerate}
\end{theorem}
\begin{proof}
(a) These judgements do not involve $\chi^i$. (b) By \cref{prop:truth}, a
for-disturbance can only enlarge each $B_j[\kappa^i]$, $j\neq i$, and an
against-disturbance only shrink it, while $A_i[\kappa^i]$ and the
improvements of $i$ are unchanged; conclude by \cref{lem:stabmono}.
(c) Follows from (b) and \cref{prop:sensitivity}.
\end{proof}

Capability and intention thus act on different parts of a stability
certificate. Assessments of capability determine which moves exist in the
perceived model. Through the focal DM's own options they affect its
improvements and hence every concept, Nash stability included; through the
opponents' options they affect the sanctioning paths and hence every
sanction-based concept, but not Nash stability. Assessments of intention
determine only which of the opponents' moves count as credible, and affect
only SEQ. The same two dimensions separate structural admissibility
from preferential credibility in the analysis of state exclusion
\cite[Sec.~6]{KatoExclusion}.

\begin{table}[H]
\centering
\small
\begin{tabular}{@{}>{\raggedright\arraybackslash}p{3.4cm}
                >{\raggedright\arraybackslash}p{4.8cm}
                >{\raggedright\arraybackslash}p{5.2cm}@{}}
\toprule
Assessment disturbed & Concepts affected & Counterpart in
\cite{KatoExclusion} \\
\midrule
Own capability ($\psi$ on $O_i$) & Nash, GMR, SMR, SEQ (available improvements and escapes)
  & --- \\
Opponents' capability ($\psi$ on $O\setminus O_i$) & GMR, SMR, SEQ; Nash unchanged (existence of sanctioning paths)
  & structural admissibility \\
Opponents' intention ($\chi$) & SEQ only (credibility of sanctioning paths)
  & preferential credibility \\
\bottomrule
\end{tabular}
\caption{Effects of capability and intention assessments on stability
concepts.}
\label{tab:dimensions}
\end{table}

For a state-preserving action $d$ of $k$, the direction of its effect on an
observer $i$ is obtained by decomposing the induced disturbance
$\kappa^i\mapsto\Gamma_i(d,\kappa^i)$ into its parts and applying
\cref{thm:direction,thm:credibility} to each.

\begin{proposition}[Judgements immune to every action]\label{prop:certificate}
Let $s\in S$, and let $\hat O(s)$ be the set of options in the requirements of
arcs of $G$ whose tail lies in $\hat D_G(s)$. If $\psi^i(o)$ is immune for
$\delta_i$ for every $o\in\hat O(s)$, and $\chi^i(a)$ is immune for $\delta_i$
for every arc $a$ of another DM with tail in $\hat D_G(s)$, then no
state-preserving action of any DM changes any of $i$'s stability
judgements at $s$.
\end{proposition}
\begin{proof}
By \cref{lem:immunewarrant} and \cref{prop:invariant}, every arc with tail in
$\hat D_G(s)$ keeps its status as perceived and as credible under every
disturbance. Perceived arcs are arcs of $G$, so the forward-reachable sets of
$s$ in the perceived models before and after the disturbance are both
contained in $\hat D_G(s)$, and they coincide with each other by induction on
path length. The localization lemma gives the claim.
\end{proof}

% ===============================================================
\section{Actions that enable provocation or deterrence}\label{sec:manip}
% ===============================================================

A state-preserving action does not change the physical state and therefore
is not an improvement in the underlying preference relation on $S$. Its
strategic value lies instead in the change it induces in the stability
judgements of others. We define this change
without assuming how the others then choose among their moves; the choice of
responses is treated in \cite{KatoAnticipation}.

\begin{definition}[Provocation- and deterrence-enabling actions]\label{def:manip}
Let $k\neq i$, $s\in S$, $d\in\Act_k(s)$, $\sigma\in\cX$, and
$\kappa'^i=\Gamma_i(d,\kappa^i)$. Write $R_i^{+}[\kappa^i](s)$ for $i$'s
improvements from $s$ in $G[\kappa^i]$.
\begin{enumerate}[label=(\alph*)]
  \item $d$ is \emph{provocation-enabling} for $k$ against $i$ at $s$ if
        $s\in\Stab^\sigma_i[\kappa^i]\setminus\Stab^\sigma_i[\kappa'^i]$ and
        some $t\in R_i^{+}[\kappa'^i](s)$ satisfies $t\succ_k s$.
  \item $d$ is \emph{deterrence-enabling} for $k$ against $i$ at $s$ if
        $s\in\Stab^\sigma_i[\kappa'^i]\setminus\Stab^\sigma_i[\kappa^i]$ and
        $s\succ_k t$ for every $t\in R_i^{+}[\kappa^i](s)$.
\end{enumerate}
\end{definition}

A provocation-enabling action turns a state that $i$ judges stable into one
it judges unstable, where at least one of $i$'s improvements is welcome to
$k$; whether $i$ takes that improvement is not asserted. A
deterrence-enabling action turns a state that $i$ judges unstable into one
it judges stable, where $k$ prefers the state to every improvement $i$ was
contemplating. The true availability of the options and the true
preferences of the DMs concerned play no role: an action works through
$i$'s interpretation, so bluffs are included.

A part of a disturbance of $i$ is \emph{effective} if applying it changes at
least one of $i$'s perceived arc sets $A_j[\kappa^i]$ or perceived credible
arc sets $B_j[\kappa^i]$.

\begin{corollary}[Channels]\label{cor:channels}
Decompose the disturbance induced by $d$ into its parts. Then:
\begin{enumerate}[label=(\alph*)]
  \item A provocation-enabling action induces at least one effective part
        that is an opponent-side against-part (of capability or intention)
        or an own-side for-part of capability; a deterrence-enabling action
        induces at least one effective part that is an opponent-side
        for-part or an own-side against-part of capability.
  \item Under $\delta_{\Nv\Yv}$ against-parts are never effective, so
        provocation can be enabled only through evidence for $i$'s own
        options, and deterrence only through evidence for its opponents'
        options or intentions.
  \item Under $\delta_{\Yv\Nv}$ for-parts are never effective, so provocation
        can be enabled only through evidence against its opponents' options
        or intentions, and deterrence only through evidence against its own
        options.
  \item Under $\delta_{\Nv\Nv}$ and $\delta_{\Yv\Yv}$ all channels are possible.
\end{enumerate}
In particular, if $k$'s actions can generate in $i$ only evidence about
options and intentions of DMs other than $i$, then an observer using
$\delta_{\Nv\Yv}$ can be subjected to deterrence-enabling but never to
provocation-enabling actions, and an observer using $\delta_{\Yv\Nv}$ the
reverse.
\end{corollary}
\begin{proof}
The updates of the parts commute on evidence pairs and lead to the same
final epistemic state, whatever the order; fix any order of application. By
\cref{thm:direction,thm:credibility}, opponent-side for-parts
and own-side against-parts can only enlarge the set of states $i$ judges
stable, and the other parts can only shrink it. Parts that are not effective
change nothing. Applying the parts one after another in the fixed order, a
state can leave the
set of states judged stable only at a step that applies an effective
shrinking part, and enter it only at a step that applies an effective
enlarging part. Parts (b) and (c) follow
from \cref{prop:sensitivity}, and (d) from \cref{ex:threat}.
\end{proof}

\begin{example}[A threat and its assessment]\label{ex:threat}
Let $N=\{1,2,3\}$ and $S=\{s,t,u\}$. DM~1 can move from $s$ to $t$
unconditionally; DM~2 can move from $t$ to $u$, a threat whose requirement is
an option $o\in O_2$; DM~3 has no moves. Preferences satisfy
$t\succ_1 s\succ_1 u$, $s\succ_2 t$ and $t\succ_3 s$; the other comparisons
play no role. The graph is shown below; both arcs are irreversible, and the
dotted arc is the one whose availability DM~1 assesses.
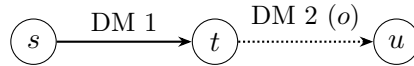
\[
\begin{tikzpicture}[x=1cm,y=1cm,>={Stealth[length=1.8mm]},
  st/.style={circle,draw,inner sep=1pt,minimum size=6mm}]
  \node[st] (s) at (0,0) {$s$};
  \node[st] (t) at (2.4,0) {$t$};
  \node[st] (u) at (4.8,0) {$u$};
  \draw[->,thick] (s) -- node[above,font=\small]{DM~1} (t);
  \draw[->,thick,densely dotted] (t) -- node[above,font=\small]{DM~2 ($o$)} (u);
\end{tikzpicture}
\]
DM~1 judges $s$
GMR-stable exactly when it perceives the threat, that is, when
$\delta_1(\psi^1(o))=\Yv$; the same holds for SMR, since DM~1 has no move from
$u$. DM~2 has an action $d_2$ (an exercise) that DM~1 reads as evidence for
$o$, and DM~3 an action $d_3$ (a leak) that DM~1 reads as evidence against.

\begin{table}[H]
\centering
\small
\begin{tabular}{@{}>{\raggedright\arraybackslash}p{3.0cm}
                >{\raggedright\arraybackslash}p{3.3cm}
                >{\raggedright\arraybackslash}p{3.6cm}
                >{\raggedright\arraybackslash}p{3.6cm}@{}}
\toprule
$\psi^1(o)$ & Action and induced change &
$\delta_1(\Bv)=\Yv$ ($\delta_{\Nv\Yv}$, $\delta_{\Yv\Yv}$) &
$\delta_1(\Bv)=\Nv$ ($\delta_{\Yv\Nv}$, $\delta_{\Nv\Nv}$) \\
\midrule
$\Nv$ (threat not believed) & $d_2$: $\Nv\to\Bv$ &
threat perceived; $d_2$ is deterrence-enabling for DM~2 &
no effect \\
$\Yv$ (threat believed) & $d_3$: $\Yv\to\Bv$ &
no effect &
threat no longer perceived; $d_3$ is provocation-enabling for DM~3 \\
\bottomrule
\end{tabular}
\caption{Effect of state-preserving actions on DM~1's judgement at $s$.}
\label{tab:threat}
\end{table}

Starting from $\Uv$, evidence for ($\Uv\to\Yv$) makes the threat perceived
under $\delta_{\Nv\Yv}$ and $\delta_{\Nv\Nv}$, and evidence against
($\Uv\to\Nv$) removes it under $\delta_{\Yv\Yv}$ and $\delta_{\Yv\Nv}$. Under
$\delta_{\Nv\Nv}$ both deterrence ($\Uv\to\Yv$) and provocation ($\Yv\to\Bv$)
can be enabled, and under $\delta_{\Yv\Yv}$ both deterrence ($\Nv\to\Bv$) and
provocation ($\Uv\to\Nv$).
\end{example}

As an implementation check of the definitions, rather than as additional
evidence for the proofs, the operator-level facts of \cref{sec:operators}
were recomputed exhaustively over $\FOUR$, and \cref{thm:direction} was tested
on 3000 randomly
generated three-DM graph models with option requirements and complete,
possibly intransitive preferences, for all four operators, all four
stability concepts and random for- and against-disturbances on both sides,
with no violation in $2.4\times10^{5}$ checks; \cref{thm:credibility,thm:hedging}
were tested on 2000 further models with no violation in $2\times10^{5}$
checks.

% ===============================================================
\section{Types and hedging}\label{sec:hedging}
% ===============================================================

Observers often assess a party not move by move but by asking which side it
is on. We model this by candidate types and derive intention assessments
from evidence about types.

\begin{definition}[Types, live types and derived intention]\label{def:types}
Let $k\neq i$. A set of \emph{candidate types} of $k$ is a finite set
$\Theta_k$ of complete relations on $S$, each a preference that $k$ might
have. A \emph{type assessment} of $i$ is a map $\tau\colon\Theta_k\to\FOUR$,
and the \emph{live types} are $L=\{\theta\in\Theta_k:\delta_i(\tau(\theta))=\Yv\}$.
The live types induce the intention assessment of $k$'s moves
\[
  \chi_L(u,v)=(f,a),\qquad
  \begin{aligned}
  f=1&\iff v\succ_\theta u \text{ for some } \theta\in L,\\
  a=1&\iff v\succ_\theta u \text{ fails for some } \theta\in L.
  \end{aligned}
\]
\end{definition}

The observer applies its operator twice: to decide which types are live, and
to decide which of $k$'s moves are credible. If no candidate type is live,
the induced intention assessment is $\Uv$ for every move of $k$, representing
the absence of any surviving basis for a judgement either way; the
observer's operator then decides how such moves are treated. Evidence against a type matters
as far as the observer's operator lets it: an assessment $\Yv$ that becomes
$\Bv$ keeps the type live under $\delta_i(\Bv)=\Yv$ and removes it under
$\delta_i(\Bv)=\Nv$. If exactly one type $\theta$ is live, we say that $k$ has
\emph{revealed} $\theta$ to $i$; if exactly two types $\theta_1,\theta_2$ are
live, $k$ is \emph{hedging} between them in the eyes of $i$.

\begin{proposition}[What hedging contests]\label{prop:hedgecontest}
If $k$ hedges between $\theta_1$ and $\theta_2$, then $\chi_L(u,v)$ is $\Yv$ if
$v$ is an improvement over $u$ under both types, $\Nv$ if under neither, and
$\Bv$ if under exactly one; the contested assessments are absorbing under
\cref{ass:accumulation}.
\end{proposition}
\begin{proof}
Immediate from \cref{def:types} with $L=\{\theta_1,\theta_2\}$ and from
\cref{prop:invariant}.
\end{proof}

\begin{theorem}[Hedging]\label{thm:hedging}
Let $i$ use $\delta_i$ and let $k\neq i$ hedge between $\theta_1$ and
$\theta_2$. Let $\mathcal{J}_{\mathrm{hedge}}$, $\mathcal{J}_{\theta_1}$ and
$\mathcal{J}_{\theta_2}$ be the sets of states $i$ judges SEQ for itself when
$k$ hedges, reveals $\theta_1$ or reveals $\theta_2$, that is, when the live
set is $\{\theta_1,\theta_2\}$, $\{\theta_1\}$ or $\{\theta_2\}$, all other
assessments being fixed.
\begin{enumerate}[label=(\alph*)]
  \item If $\delta_i(\Bv)=\Yv$, $k$'s perceived credible arcs under hedging are
        the union of those under the two revelations, and hedging weakly
        expands $i$'s set of SEQ-stable states:
        $\mathcal{J}_{\theta_1}\cup\mathcal{J}_{\theta_2}\subseteq
        \mathcal{J}_{\mathrm{hedge}}$.
  \item If $\delta_i(\Bv)=\Nv$, they are the intersection, and hedging weakly
        shrinks it: $\mathcal{J}_{\mathrm{hedge}}\subseteq
        \mathcal{J}_{\theta_1}\cap\mathcal{J}_{\theta_2}$.
\end{enumerate}
\end{theorem}
\begin{proof}
By \cref{prop:hedgecontest}, the intention assessments under hedging and
under the revelations differ only on disputed moves, which are $\Bv$ under
hedging and $\Yv$ under one revelation and $\Nv$ under the other. The
statements on arcs follow, and the inclusions follow from the monotonicity
of SEQ judgements in credible arcs (\cref{lem:stabmono}).
\end{proof}

\Cref{thm:hedging} compares the observer's judgements, not the hedger's
payoffs, and makes no assumption about $k$'s true preferences.

\begin{example}[Two powers and a hedging party]\label{ex:hedge}
There are three DMs, $N=\{P_1,P_2,H\}$, and five states,
$S=\{s,t_1,t_2,u_1,u_2\}$, with status quo $s$. Power $P_1$ can move from $s$ to $t_1$ and power $P_2$ from $s$ to $t_2$, each
an encroachment it prefers. The third party $H$ can respond to $t_1$ by
moving to $u_1$, siding with $P_2$, and to $t_2$ by moving to $u_2$, siding
with $P_1$. Preferences are
\[
t_1\succ_{P_1}s\succ_{P_1}t_2\succ_{P_1}u_2\succ_{P_1}u_1,\qquad
t_2\succ_{P_2}s\succ_{P_2}t_1\succ_{P_2}u_1\succ_{P_2}u_2,
\]
and $s\succ_H t_1\succ_H t_2\succ_H u_1\succ_H u_2$: $H$ ranks $s$ first and
prefers staying at $t_1$ and $t_2$ to siding with either power. The graph is
shown below; all arcs are irreversible.
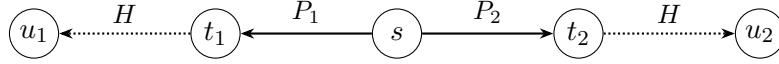
\[
\begin{tikzpicture}[x=1cm,y=1cm,>={Stealth[length=1.8mm]},
  st/.style={circle,draw,inner sep=1pt,minimum size=6.5mm}]
  \node[st] (s) at (0,0) {$s$};
  \node[st] (t1) at (-2.4,0) {$t_1$};
  \node[st] (t2) at (2.4,0) {$t_2$};
  \node[st] (u1) at (-4.8,0) {$u_1$};
  \node[st] (u2) at (4.8,0) {$u_2$};
  \draw[->,thick] (s) -- node[above,font=\small]{$P_1$} (t1);
  \draw[->,thick] (s) -- node[above,font=\small]{$P_2$} (t2);
  \draw[->,thick,densely dotted] (t1) -- node[above,font=\small]{$H$} (u1);
  \draw[->,thick,densely dotted] (t2) -- node[above,font=\small]{$H$} (u2);
\end{tikzpicture}
\]
The candidate types of $H$ are $\theta_A$, under which siding with $P_1$ is an
improvement and siding with $P_2$ is not, and $\theta_B$, the reverse.

\begin{table}[H]
\centering
\small
\begin{tabular}{@{}lcc@{}}
\toprule
& observers with $\delta(\Bv)=\Yv$ & observers with $\delta(\Bv)=\Nv$ \\
\midrule
$H$ reveals $\theta_A$ & $P_2$ judges $s$ SEQ, $P_1$ not & $P_2$ judges $s$ SEQ, $P_1$ not \\
$H$ reveals $\theta_B$ & $P_1$ judges $s$ SEQ, $P_2$ not & $P_1$ judges $s$ SEQ, $P_2$ not \\
$H$ hedges & both judge $s$ SEQ & neither \\
\bottomrule
\end{tabular}
\caption{SEQ judgements of the status quo in \cref{ex:hedge}.}
\label{tab:hedge}
\end{table}

Against observers who read contradiction as credibility, hedging makes both
powers judge the status quo SEQ-stable, although in the true model it is
SEQ-stable for neither, since $H$'s responses are not improvements for $H$.
It is GMR-stable for both in the true model, because GMR asks only whether
$H$ can respond, not whether it would.
\end{example}

% ===============================================================
\section{Illustration: the 1995 DVD format negotiation}\label{sec:dvd}
% ===============================================================

\subsection{Background}

In late 1994 two incompatible high-density optical disc formats were
announced: the Multimedia Compact Disc (MMCD) of Philips and Sony, and the
Super Density (SD) disc of Toshiba, Matsushita, Time Warner and others
\cite{Taylor1998,Wolpin2015}. On 3 May 1995, technical representatives of
Apple, Compaq, Hewlett-Packard, IBM and Microsoft publicly urged the two
camps to converge on a single format, reporting that they had met each camp
separately in April and urged it to present one \cite{Apple1995}. The group,
in which IBM played a prominent role, was open to further computer
companies, and a near-contemporaneous account names additional members
\cite{Apple1995,Doyle1996}; we refer to it as the Technical Working Group
(TWG). Both camps regarded its support as critical for the adoption of
their format in personal computers \cite{Wolpin2015}. A general agreement on the basic specifications of a single
format was reached on 15 September 1995
\cite{Nokia1995,Toshiba1995,UPI1995,LAT1995}, and the final details were
announced on 8 December 1995 \cite{Toshiba1995}.

The accounts agree that the TWG, with IBM prominent, exerted pressure on
both camps, but they differ on the signalling in the summer of 1995. A later
retrospective reports that by late July IBM let the camps know that it was
ready to endorse SD, and that the SD camp, which regarded its format as the
likely winner, was reluctant until IBM informed both camps in September that
it would endorse the compromise \cite{Wolpin2015}. A near-contemporaneous
account reports audience-specific communications: Sony and Philips were
told that IBM was about to adopt the longer SD format and were given a few
weeks to reach a compromise, while Toshiba's designers were told that
Sony's modulation and storage techniques were preferable \cite{Doyle1996}.

We use the case to illustrate the distinction between capability and
intention, not to establish the causes of the outcome. As shown below, the
framework accommodates either account of the signalling.

\subsection{Model}

Following the usual presentation of a graph model, we specify the decision
makers and their options, the feasible states, the transitions with their
reversibility, and the preferences.

\paragraph{Decision makers and options.}
The DMs are the MMCD camp $M$, the SD camp $S$ and the TWG $T$, with the
options listed in \cref{tab:dvdoptions}. The TWG's public position was a call
for a single format \cite{Apple1995}, and its formal stance was initially to
endorse neither format and to issue non-binding recommendations; its
influence rested on the de facto market power of its members
\cite{Wolpin2015}. The option $B$ stylises this as an explicit refusal of
support for a format launched by one camp alone; it is not a literal
boycott, which the sources do not report.

\begin{table}[H]
\centering
\small
\begin{tabular}{@{}lll@{}}
\toprule
DM & Option & Meaning \\
\midrule
$M$ & $L_M$ & launch the MMCD format unilaterally \\
    & $A_M$ & accept a unified format \\
$S$ & $L_S$ & launch the SD format unilaterally \\
    & $A_S$ & accept a unified format \\
$T$ & $E_M$ & endorse MMCD \\
    & $E_S$ & endorse SD \\
    & $B$   & refuse support to any single format launched unilaterally \\
    & $E_C$ & endorse a compromise accepted by a camp \\
\bottomrule
\end{tabular}
\caption{Decision makers and options of the DVD model.}
\label{tab:dvdoptions}
\end{table}

\paragraph{Feasible states.}
A camp takes at most one of its options: accepting a unified format
supersedes a launch. The TWG takes at most one of its options. In addition,
$E_C$ requires that some camp has accepted, and $T$ does not endorse a format
whose own camp has accepted a unified format. When both camps have
accepted, the format is unified and the conflict is resolved; all such
configurations are identified with a single state, in which no statement of
$T$ is recorded. These constraints leave $33$ feasible states, numbered as in
\cref{fig:dvdgraph} and listed in option form in \cref{tab:dvdstates}
(\cref{app:dvd}). The status quo is state~$1$, in which no option is taken,
and the unified format is state~$33$.

\paragraph{Transitions and reversibility.}
\Cref{fig:dvdoptions} shows the option-level transitions of each DM. A camp
can launch, or accept directly, from the status quo, and can accept after
launching. The TWG can make one statement, while the format is not unified.
All transitions are irreversible: no launch, acceptance or statement can be
withdrawn. \Cref{fig:dvdgraph} shows the resulting integrated graph on the
33 feasible states, with $88$ arcs.

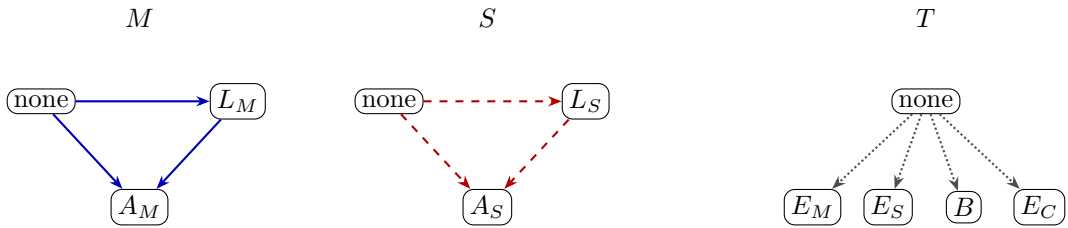
\begin{figure}[H]
\centering
\begin{tikzpicture}[x=1cm,y=1cm,>={Stealth[length=1.8mm]},
  op/.style={draw,rounded corners,inner sep=2.5pt,font=\small}]
\begin{scope}
  \node[font=\small] at (1.3,1.1) {$M$};
  \node[op] (m0) at (0,0) {none};
  \node[op] (m1) at (2.6,0) {$L_M$};
  \node[op] (m2) at (1.3,-1.4) {$A_M$};
  \draw[->,thick,blue!70!black] (m0)--(m1);
  \draw[->,thick,blue!70!black] (m0)--(m2);
  \draw[->,thick,blue!70!black] (m1)--(m2);
\end{scope}
\begin{scope}[shift={(4.6,0)}]
  \node[font=\small] at (1.3,1.1) {$S$};
  \node[op] (s0) at (0,0) {none};
  \node[op] (s1) at (2.6,0) {$L_S$};
  \node[op] (s2) at (1.3,-1.4) {$A_S$};
  \draw[->,thick,red!70!black,dashed] (s0)--(s1);
  \draw[->,thick,red!70!black,dashed] (s0)--(s2);
  \draw[->,thick,red!70!black,dashed] (s1)--(s2);
\end{scope}
\begin{scope}[shift={(10.2,0)}]
  \node[font=\small] at (1.5,1.1) {$T$};
  \node[op] (t0) at (1.5,0) {none};
  \node[op] (t1) at (0,-1.4) {$E_M$};
  \node[op] (t2) at (1,-1.4) {$E_S$};
  \node[op] (t3) at (2,-1.4) {$B$};
  \node[op] (t4) at (3,-1.4) {$E_C$};
  \foreach \k in {1,2,3,4} \draw[->,thick,black!70,densely dotted] (t0)--(t\k);
\end{scope}
\end{tikzpicture}
\caption{Option-level transitions of the three DMs. Every arrow is one-way:
all transitions are irreversible. $E_M$ ($E_S$) is unavailable once $M$
($S$) has accepted, $E_C$ requires that a camp has accepted, and no DM moves
once the format is unified.}
\label{fig:dvdoptions}
\end{figure}

\begin{figure}[p]
\centering
\pgfdeclarelayer{bg}\pgfsetlayers{bg,main}
\begin{tikzpicture}[x=1cm,y=1cm,>={Stealth[length=1.6mm]},st/.style={circle,draw,inner sep=0.5pt,minimum size=5.2mm,font=\scriptsize,fill=white}]
\node[font=\small] at (0.0,1.9) {$M$: none};
\node[font=\small] at (5.3,1.9) {$M$: launched};
\node[font=\small] at (10.6,1.9) {$M$: accepted};
\node[font=\small,rotate=90] at (-2.3,0.0) {$S$: none};
\node[font=\small,rotate=90] at (-2.3,-4.2) {$S$: launched};
\node[font=\small,rotate=90] at (-2.3,-8.4) {$S$: accepted};
\draw[gray!50,rounded corners] (-1.85,-1.45) rectangle (1.85,1.45);
\draw[gray!50,rounded corners] (-1.85,-5.65) rectangle (1.85,-2.75);
\draw[gray!50,rounded corners] (-1.85,-9.85) rectangle (1.85,-6.95);
\draw[gray!50,rounded corners] (3.4499999999999997,-1.45) rectangle (7.15,1.45);
\draw[gray!50,rounded corners] (3.4499999999999997,-5.65) rectangle (7.15,-2.75);
\draw[gray!50,rounded corners] (3.4499999999999997,-9.85) rectangle (7.15,-6.95);
\draw[gray!50,rounded corners] (8.75,-1.45) rectangle (12.45,1.45);
\draw[gray!50,rounded corners] (8.75,-5.65) rectangle (12.45,-2.75);
\draw[gray!50,rounded corners] (8.75,-9.85) rectangle (12.45,-6.95);
\node[st] (n1) at (0.00,0.00) {1};
\node[st] (n2) at (-1.30,0.95) {2};
\node[st] (n3) at (1.30,0.45) {3};
\node[st] (n4) at (-0.55,-0.95) {4};
\node[st] (n5) at (0.00,-4.20) {5};
\node[st] (n6) at (-1.30,-3.25) {6};
\node[st] (n7) at (1.30,-3.75) {7};
\node[st] (n8) at (-0.55,-5.15) {8};
\node[st] (n9) at (0.00,-8.40) {9};
\node[st] (n10) at (-1.30,-7.45) {10};
\node[st] (n11) at (-0.55,-9.35) {11};
\node[st] (n12) at (0.70,-9.00) {12};
\node[st] (n13) at (5.30,0.00) {13};
\node[st] (n14) at (4.00,0.95) {14};
\node[st] (n15) at (6.60,0.45) {15};
\node[st] (n16) at (4.75,-0.95) {16};
\node[st] (n17) at (5.30,-4.20) {17};
\node[st] (n18) at (4.00,-3.25) {18};
\node[st] (n19) at (6.60,-3.75) {19};
\node[st] (n20) at (4.75,-5.15) {20};
\node[st] (n21) at (5.30,-8.40) {21};
\node[st] (n22) at (4.00,-7.45) {22};
\node[st] (n23) at (4.75,-9.35) {23};
\node[st] (n24) at (6.00,-9.00) {24};
\node[st] (n25) at (10.60,0.00) {25};
\node[st] (n26) at (11.90,0.45) {26};
\node[st] (n27) at (10.05,-0.95) {27};
\node[st] (n28) at (11.30,-0.60) {28};
\node[st] (n29) at (10.60,-4.20) {29};
\node[st] (n30) at (11.90,-3.75) {30};
\node[st] (n31) at (10.05,-5.15) {31};
\node[st] (n32) at (11.30,-4.80) {32};
\node[st] (n33) at (10.60,-8.40) {33};
\begin{pgfonlayer}{bg}
\draw[->,black!70,densely dotted,thick] (n1) to (n2);
\draw[->,black!70,densely dotted,thick] (n1) to (n3);
\draw[->,black!70,densely dotted,thick] (n1) to (n4);
\draw[->,black!70,densely dotted,thick] (n5) to (n6);
\draw[->,black!70,densely dotted,thick] (n5) to (n7);
\draw[->,black!70,densely dotted,thick] (n5) to (n8);
\draw[->,black!70,densely dotted,thick] (n9) to (n10);
\draw[->,black!70,densely dotted,thick] (n9) to (n11);
\draw[->,black!70,densely dotted,thick] (n9) to (n12);
\draw[->,black!70,densely dotted,thick] (n13) to (n14);
\draw[->,black!70,densely dotted,thick] (n13) to (n15);
\draw[->,black!70,densely dotted,thick] (n13) to (n16);
\draw[->,black!70,densely dotted,thick] (n17) to (n18);
\draw[->,black!70,densely dotted,thick] (n17) to (n19);
\draw[->,black!70,densely dotted,thick] (n17) to (n20);
\draw[->,black!70,densely dotted,thick] (n21) to (n22);
\draw[->,black!70,densely dotted,thick] (n21) to (n23);
\draw[->,black!70,densely dotted,thick] (n21) to (n24);
\draw[->,black!70,densely dotted,thick] (n25) to (n26);
\draw[->,black!70,densely dotted,thick] (n25) to (n27);
\draw[->,black!70,densely dotted,thick] (n25) to (n28);
\draw[->,black!70,densely dotted,thick] (n29) to (n30);
\draw[->,black!70,densely dotted,thick] (n29) to (n31);
\draw[->,black!70,densely dotted,thick] (n29) to (n32);
\draw[->,blue!70!black,thick] (n1) to (n13);
\draw[->,blue!70!black,thick,bend left=16,opacity=0.75] (n1) to (n25);
\draw[->,blue!70!black,thick] (n2) to (n14);
\draw[->,blue!70!black,thick] (n3) to (n15);
\draw[->,blue!70!black,thick,bend left=16,opacity=0.75] (n3) to (n26);
\draw[->,blue!70!black,thick] (n4) to (n16);
\draw[->,blue!70!black,thick,bend left=16,opacity=0.75] (n4) to (n27);
\draw[->,blue!70!black,thick] (n5) to (n17);
\draw[->,blue!70!black,thick,bend left=16,opacity=0.75] (n5) to (n29);
\draw[->,blue!70!black,thick] (n6) to (n18);
\draw[->,blue!70!black,thick] (n7) to (n19);
\draw[->,blue!70!black,thick,bend left=16,opacity=0.75] (n7) to (n30);
\draw[->,blue!70!black,thick] (n8) to (n20);
\draw[->,blue!70!black,thick,bend left=16,opacity=0.75] (n8) to (n31);
\draw[->,blue!70!black,thick] (n9) to (n21);
\draw[->,blue!70!black,thick,bend left=16,opacity=0.75] (n9) to (n33);
\draw[->,blue!70!black,thick] (n10) to (n22);
\draw[->,blue!70!black,thick,bend left=16,opacity=0.75] (n10) to (n33);
\draw[->,blue!70!black,thick] (n11) to (n23);
\draw[->,blue!70!black,thick,bend left=16,opacity=0.75] (n11) to (n33);
\draw[->,blue!70!black,thick] (n12) to (n24);
\draw[->,blue!70!black,thick,bend left=16,opacity=0.75] (n12) to (n33);
\draw[->,blue!70!black,thick] (n13) to (n25);
\draw[->,blue!70!black,thick] (n15) to (n26);
\draw[->,blue!70!black,thick] (n16) to (n27);
\draw[->,blue!70!black,thick] (n17) to (n29);
\draw[->,blue!70!black,thick] (n19) to (n30);
\draw[->,blue!70!black,thick] (n20) to (n31);
\draw[->,blue!70!black,thick] (n21) to (n33);
\draw[->,blue!70!black,thick] (n22) to (n33);
\draw[->,blue!70!black,thick] (n23) to (n33);
\draw[->,blue!70!black,thick] (n24) to (n33);
\draw[->,red!70!black,thick,dashed] (n1) to (n5);
\draw[->,red!70!black,thick,dashed,bend right=16,opacity=0.75] (n1) to (n9);
\draw[->,red!70!black,thick,dashed] (n2) to (n6);
\draw[->,red!70!black,thick,dashed,bend right=16,opacity=0.75] (n2) to (n10);
\draw[->,red!70!black,thick,dashed] (n3) to (n7);
\draw[->,red!70!black,thick,dashed] (n4) to (n8);
\draw[->,red!70!black,thick,dashed,bend right=16,opacity=0.75] (n4) to (n11);
\draw[->,red!70!black,thick,dashed] (n5) to (n9);
\draw[->,red!70!black,thick,dashed] (n6) to (n10);
\draw[->,red!70!black,thick,dashed] (n8) to (n11);
\draw[->,red!70!black,thick,dashed] (n13) to (n17);
\draw[->,red!70!black,thick,dashed,bend right=16,opacity=0.75] (n13) to (n21);
\draw[->,red!70!black,thick,dashed] (n14) to (n18);
\draw[->,red!70!black,thick,dashed,bend right=16,opacity=0.75] (n14) to (n22);
\draw[->,red!70!black,thick,dashed] (n15) to (n19);
\draw[->,red!70!black,thick,dashed] (n16) to (n20);
\draw[->,red!70!black,thick,dashed,bend right=16,opacity=0.75] (n16) to (n23);
\draw[->,red!70!black,thick,dashed] (n17) to (n21);
\draw[->,red!70!black,thick,dashed] (n18) to (n22);
\draw[->,red!70!black,thick,dashed] (n20) to (n23);
\draw[->,red!70!black,thick,dashed] (n25) to (n29);
\draw[->,red!70!black,thick,dashed,bend right=16,opacity=0.75] (n25) to (n33);
\draw[->,red!70!black,thick,dashed] (n26) to (n30);
\draw[->,red!70!black,thick,dashed,bend right=16,opacity=0.75] (n26) to (n33);
\draw[->,red!70!black,thick,dashed] (n27) to (n31);
\draw[->,red!70!black,thick,dashed,bend right=16,opacity=0.75] (n27) to (n33);
\draw[->,red!70!black,thick,dashed] (n28) to (n32);
\draw[->,red!70!black,thick,dashed,bend right=16,opacity=0.75] (n28) to (n33);
\draw[->,red!70!black,thick,dashed] (n29) to (n33);
\draw[->,red!70!black,thick,dashed] (n30) to (n33);
\draw[->,red!70!black,thick,dashed] (n31) to (n33);
\draw[->,red!70!black,thick,dashed] (n32) to (n33);
\end{pgfonlayer}
\begin{scope}[shift={(0,-10.9)}]
\draw[->,blue!70!black,thick] (-1.8,0)--(-0.8,0) node[right,font=\small,black]{move of $M$};
\draw[->,red!70!black,thick,dashed] (2.6,0)--(3.6,0) node[right,font=\small,black]{move of $S$};
\draw[->,black!70,densely dotted,thick] (7.0,0)--(8.0,0) node[right,font=\small,black]{move of $T$};
\end{scope}
\end{tikzpicture}
\caption{Integrated graph of the DVD model. States are grouped by the
positions of the two camps (columns: $M$; rows: $S$); within a group, the
centre position is the state in which $T$ has made no statement, and the
other positions correspond to $E_M$ (upper left), $E_S$ (right), $B$ (lower
left) and $E_C$ (lower right). All arcs are irreversible. The status quo is
state~$1$ and the unified format state~$33$; the option form of every state
is given in \cref{tab:dvdstates}.}
\label{fig:dvdgraph}
\end{figure}
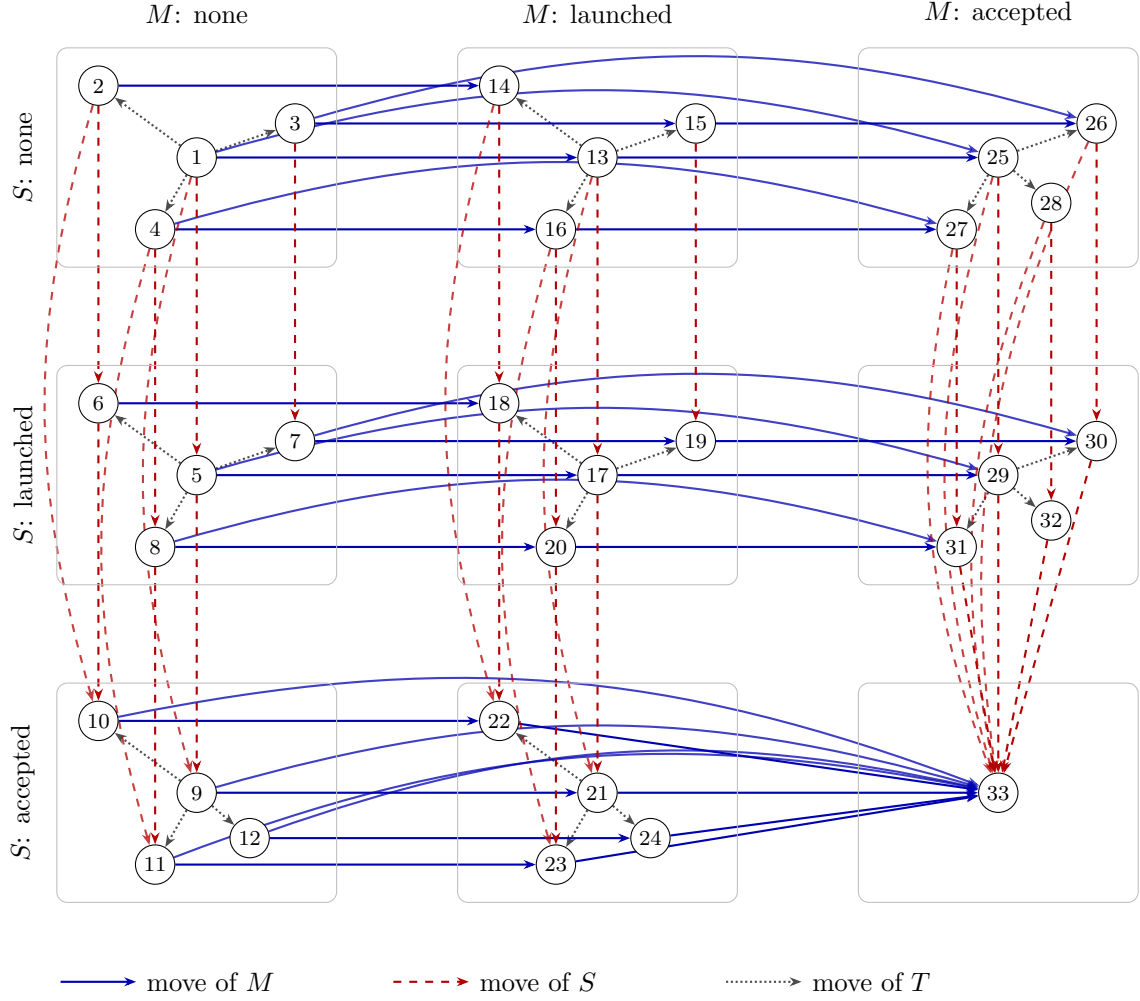

\paragraph{Preferences.}
Each camp most prefers to launch its format with the TWG's endorsement and
least prefers the rival's endorsed launch; unification lies in between and
above the status quo; a launch without endorsement is valued above the
status quo by the camp that launches (first-mover hope) but below it if the TWG
refuses support. The complete rankings, which are weak orders, are given in
\cref{tab:dvdstates}; the principles from which they are constructed, and
the comparisons on which the results depend, are set out in
\cref{app:prefbasis}. The TWG's capability to endorse or to refuse support is not in
doubt; what the camps assess is its intention, through five candidate types,
whose rankings are also given in \cref{tab:dvdstates}: $\theta_B$
(\emph{blocker}: refuses support to any single launch and otherwise does nothing),
$\theta_U$ (\emph{enforcer}: refuses support to any single launch and endorses a
compromise once one camp has accepted it), $\theta_{US}$ (an enforcer that
endorses SD if the stalemate persists), and the partisans $\theta_S$ and
$\theta_M$. The analyst takes $T$'s true preferences to be those of
$\theta_U$.

\paragraph{The arcs that matter.}
The analysis below uses few of the arcs. Write $\ell_M$ (state~$13$) for the
state in which $M$ has launched and nobody else has moved, and $b_M$, $e^S_M$
and $e^M_M$ (states $16$, $15$ and $14$) for the states reached from $\ell_M$
when $T$ refuses support, endorses SD or endorses MMCD; $\ell_S$, $b_S$, $e^M_S$ and
$e^S_S$ (states $5$, $8$, $6$ and $7$) are defined symmetrically. The relevant
arcs are
\[
  s_0\xrightarrow{M}\ell_M,\qquad
  \ell_M\xrightarrow{T}b_M,\quad
  \ell_M\xrightarrow{T}e^S_M,\quad
  \ell_M\xrightarrow{T}e^M_M,
\]
and symmetrically for $S$, with
$\ell_M\succ_M s_0\succ_M b_M,\,e^S_M$ and $e^M_M\succ_M s_0$. Thus $M$'s
launch is an improvement; it has a GMR-sanction under every assessment,
because $T$ can move to $b_M$ or $e^S_M$, both of which satisfy
$s_0\succ_M\cdot$; and it has an SEQ-sanction exactly when one of these
sanctioning moves is also an improvement for $T$ under a live type: under
$\theta_B$, $\theta_U$ and $\theta_{US}$ (refusal) and under $\theta_S$
(endorsement of SD), but not under $\theta_M$, whose best response is
$e^M_M$. The same holds for $S$ with the roles of the partisans exchanged.
The rival camp's responses do not change this: launching its own format in
reply starts a format war, which is a GMR-sanction but not an SEQ-sanction,
since it is not an improvement for the camp that would start it; accepting
after the rival's launch does not lower the launcher's position below the
status quo and is no sanction at all.

\subsection{Capability and intention in the case}

\Cref{tab:dvdstatic} shows whether each camp judges the status quo stable
against its own unilateral launch, under SEQ and under GMR, for various sets
of live TWG types.

\begin{table}[H]
\centering
\small
\begin{tabular}{@{}lcc@{}}
\toprule
Live TWG types & $M$ judges $s_0$ stable (SEQ / GMR) & $S$ judges $s_0$ stable (SEQ / GMR) \\
\midrule
$\{\theta_B\}$, $\{\theta_U\}$ or $\{\theta_{US}\}$ & yes / yes & yes / yes \\
$\{\theta_S\}$ & yes / yes & \textbf{no} / yes \\
$\{\theta_M\}$ & \textbf{no} / yes & yes / yes \\
$\{\theta_M,\theta_S\}$ & yes / yes & yes / yes \\
\bottomrule
\end{tabular}
\caption{Stability judgements of the status quo $s_0$ against unilateral
launch.}
\label{tab:dvdstatic}
\end{table}

Under GMR every camp judges the status quo stable under every assessment,
because the TWG can always endorse the rival format or refuse support; GMR asks
whether the TWG \emph{can} sanction and cannot distinguish the phases of the
negotiation. Under SEQ the judgement depends on whether the sanctioning move
is one the TWG \emph{would} make, that is, on the assessment of its
intention. The row $\{\theta_M,\theta_S\}$ is an instance of
\cref{thm:hedging}: a TWG that appears as a possible partisan of either side
deters both, for observers who read contradiction as credibility.

\subsection{State-preserving actions and their interpretations}

The TWG's public statement of May 1995, its communications in the summer
and its September communication change nothing in the physical state: no
format is launched, endorsed or accepted by them. They are state-preserving
actions of $T$, and the sources suggest the following interpretations.
\begin{itemize}[leftmargin=*]
  \item \emph{Public call for a single format} (May 1995): read by both camps
        as evidence for $\theta_B$. Both camps judge the status quo
        SEQ-stable; neither launches.
  \item \emph{Summer communications}. Under the later account, a common
        action $d$ observed by both camps is read differently:
        $\Gamma_M(d,\kappa^M)$ generates evidence for $\theta_{US}$, an
        enforcer with an SD fallback, and $\Gamma_S(d,\kappa^S)$ evidence for
        $\theta_S$, partisan support. Under the near-contemporaneous
        account, $T$ performs distinct actions $d_M\neq d_S$: $d_M$, observed
        by $M$ only, generates evidence for $\theta_{US}$ in $M$, while $d_S$
        conveys technical information to $S$, whose reading of $T$ as a likely
        supporter of SD rests on its own expectation of victory. Either way,
        the readings are $\{\theta_{US}\}$ for $M$ and $\{\theta_S\}$ for $S$.
        With $\{\theta_S\}$ live, $S$ no longer judges the status quo
        SEQ-stable against its own launch (\cref{tab:dvdstatic}): its launch
        has no credible sanction, consistent with its regarding SD as the
        likely winner. $M$ continues to judge the status quo stable.
  \item \emph{September communication}: an action observed by $S$ that makes
        an enforcer type live for it; $S$ again judges the status quo
        SEQ-stable against its own launch.
\end{itemize}
The two accounts of the summer correspond to the two ways in which
\cref{def:action} lets the effect of actions differ across observers:
observer-specific interpretation of a common action, and selective
disclosure through distinct actions. Why $M$ then accepted a unified
format, which is worse for it than the status quo in itself, requires an
account of how DMs anticipate the evolution of the conflict; it is given in
\cite{KatoAnticipation}.

\subsection{Limitations}

The dates and the content of the public statements rest on primary sources
\cite{Apple1995,Nokia1995,Toshiba1995} and contemporary reports
\cite{UPI1995,LAT1995}. The private communications of the summer of 1995 are
known only from a near-contemporaneous account, which reports them at second
hand \cite{Doyle1996}, and from a later retrospective \cite{Wolpin2015}. The
preferences (\cref{app:prefbasis}) and the readings are reconstructions. The Hollywood studios,
whose advisory committee had set requirements of its own
\cite{Wolpin2015}, are left out, and each camp and the TWG are treated as
single DMs.

% ===============================================================
\section{Discussion}
% ===============================================================

\paragraph{Interpretation.}
The four reduction operators become profiles of vulnerability to
manipulation. A DM that reads any evidence for an option as availability
($\delta_{\Nv\Yv}$) can be made more cautious but cannot be emboldened by
talk about its opponents; a DM that reads any evidence against as
unavailability ($\delta_{\Yv\Nv}$) can be emboldened but not made more
cautious; the other two can be moved either way. Under monotone
accumulation, contradiction is absorbing, so injecting contradictory
evidence is effective only once and the ambiguity created by hedging
persists. A party that makes its capabilities ambiguous acts on every
sanction-based judgement of its opponents; a party that makes only its
intentions ambiguous acts on credibility alone.

\paragraph{Limitations and further work.}
Relating the conditions under which manipulation of preferences fails
\cite{Inohara1997} to those under which manipulation of assessments succeeds
is an open question.
Evidence is assumed to accumulate monotonically; retraction, correction and
the discrediting of sources would allow contradictory assessments to be
resolved and require a different treatment. Physical moves are assumed to
leave assessments unchanged. Higher-order assessments (what a DM believes
about another's assessment) and coalitional stability
\cite{Kilgour2001,Inohara2008} are not treated; for coalitions, disturbances
raise the questions whether members share assessments and whether an action
can target a coalition as a whole. A third epistemic layer, in which a
state-preserving action is perceived by others as a change of physical
state, would model feints in the narrow sense. The distinction between
inaction and state-preserving actions also suggests a categorical treatment
in which the latter are non-identity endomorphisms \cite{KatoCGMCR}. How DMs
anticipate the evolution of the conflict under their assessments, and how an
actor can plan a sequence of state-preserving actions, is developed in
\cite{KatoAnticipation}. Finally, all results are finite and relational,
which makes them natural targets for formal verification alongside the
existing Lean development of \cite{KatoQCW}.

% ===============================================================
\appendix
% ===============================================================
\section{States and preferences of the DVD model}\label{app:dvd}
% ===============================================================

\subsection{States and rankings}

\Cref{tab:dvdstates} lists the feasible states of the DVD model in option
form and gives each DM's ranking: $1$ is the most preferred, and states with
the same number are equally preferred. The columns for $T$ give the rankings
of the five candidate types; the true type is $\theta_U$.

{\footnotesize
\setlength{\tabcolsep}{3.2pt}
\begin{longtable}{@{}r|cc|cc|cccc|cc|ccccc@{}}
\caption{Feasible states (option form) and preference rankings of the DVD
model. Y: the option is taken; --: it is not.}\label{tab:dvdstates}\\
\toprule
& \multicolumn{2}{c|}{$M$} & \multicolumn{2}{c|}{$S$} & \multicolumn{4}{c|}{$T$} &
\multicolumn{2}{c|}{camps} & \multicolumn{5}{c}{types of $T$}\\
State & $L_M$ & $A_M$ & $L_S$ & $A_S$ & $E_M$ & $E_S$ & $B$ & $E_C$ &
$M$ & $S$ & $\theta_B$ & $\theta_U$ & $\theta_{US}$ & $\theta_S$ & $\theta_M$\\
\midrule
\endfirsthead
\toprule
State & $L_M$ & $A_M$ & $L_S$ & $A_S$ & $E_M$ & $E_S$ & $B$ & $E_C$ &
$M$ & $S$ & $\theta_B$ & $\theta_U$ & $\theta_{US}$ & $\theta_S$ & $\theta_M$\\
\midrule
\endhead
1 & -- & -- & -- & -- & -- & -- & -- & -- & 6 & 6 & 3 & 4 & 4 & 4 & 4 \\
2 & -- & -- & -- & -- & Y & -- & -- & -- & 6 & 6 & 4 & 5 & 5 & 6 & 2 \\
3 & -- & -- & -- & -- & -- & Y & -- & -- & 6 & 6 & 4 & 5 & 3 & 2 & 6 \\
4 & -- & -- & -- & -- & -- & -- & Y & -- & 6 & 6 & 3 & 4 & 4 & 5 & 5 \\
5 & -- & -- & Y & -- & -- & -- & -- & -- & 8 & 5 & 5 & 6 & 6 & 5 & 5 \\
6 & -- & -- & Y & -- & Y & -- & -- & -- & 5 & 10 & 6 & 7 & 7 & 6 & 2 \\
7 & -- & -- & Y & -- & -- & Y & -- & -- & 12 & 1 & 6 & 7 & 7 & 1 & 6 \\
8 & -- & -- & Y & -- & -- & -- & Y & -- & 6 & 9 & 2 & 3 & 3 & 5 & 5 \\
9 & -- & -- & -- & Y & -- & -- & -- & -- & 3 & 7 & 3 & 4 & 4 & 4 & 4 \\
10 & -- & -- & -- & Y & Y & -- & -- & -- & 2 & 10 & 4 & 5 & 5 & 6 & 2 \\
11 & -- & -- & -- & Y & -- & -- & Y & -- & 3 & 7 & 3 & 4 & 4 & 5 & 5 \\
12 & -- & -- & -- & Y & -- & -- & -- & Y & 10 & 3 & 4 & 2 & 2 & 5 & 5 \\
13 & Y & -- & -- & -- & -- & -- & -- & -- & 5 & 8 & 5 & 6 & 6 & 5 & 5 \\
14 & Y & -- & -- & -- & Y & -- & -- & -- & 1 & 12 & 6 & 7 & 7 & 6 & 1 \\
15 & Y & -- & -- & -- & -- & Y & -- & -- & 10 & 5 & 6 & 7 & 7 & 2 & 6 \\
16 & Y & -- & -- & -- & -- & -- & Y & -- & 9 & 6 & 2 & 3 & 3 & 5 & 5 \\
17 & Y & -- & Y & -- & -- & -- & -- & -- & 10 & 10 & 5 & 6 & 6 & 5 & 5 \\
18 & Y & -- & Y & -- & Y & -- & -- & -- & 5 & 12 & 6 & 7 & 7 & 6 & 1 \\
19 & Y & -- & Y & -- & -- & Y & -- & -- & 12 & 5 & 6 & 7 & 7 & 1 & 6 \\
20 & Y & -- & Y & -- & -- & -- & Y & -- & 11 & 11 & 2 & 3 & 3 & 5 & 5 \\
21 & Y & -- & -- & Y & -- & -- & -- & -- & 5 & 8 & 5 & 6 & 6 & 5 & 5 \\
22 & Y & -- & -- & Y & Y & -- & -- & -- & 1 & 12 & 6 & 7 & 7 & 6 & 1 \\
23 & Y & -- & -- & Y & -- & -- & Y & -- & 9 & 6 & 2 & 3 & 3 & 5 & 5 \\
24 & Y & -- & -- & Y & -- & -- & -- & Y & 10 & 5 & 6 & 7 & 7 & 5 & 5 \\
25 & -- & Y & -- & -- & -- & -- & -- & -- & 7 & 3 & 3 & 4 & 4 & 4 & 4 \\
26 & -- & Y & -- & -- & -- & Y & -- & -- & 10 & 2 & 4 & 5 & 5 & 2 & 6 \\
27 & -- & Y & -- & -- & -- & -- & Y & -- & 7 & 3 & 3 & 4 & 4 & 5 & 5 \\
28 & -- & Y & -- & -- & -- & -- & -- & Y & 3 & 10 & 4 & 2 & 2 & 5 & 5 \\
29 & -- & Y & Y & -- & -- & -- & -- & -- & 8 & 5 & 5 & 6 & 6 & 5 & 5 \\
30 & -- & Y & Y & -- & -- & Y & -- & -- & 12 & 1 & 6 & 7 & 7 & 1 & 6 \\
31 & -- & Y & Y & -- & -- & -- & Y & -- & 6 & 9 & 2 & 3 & 3 & 5 & 5 \\
32 & -- & Y & Y & -- & -- & -- & -- & Y & 5 & 10 & 6 & 7 & 7 & 5 & 5 \\
33 & -- & Y & -- & Y & -- & -- & -- & -- & 4 & 4 & 1 & 1 & 1 & 3 & 3 \\
\bottomrule
\end{longtable}
}

\subsection{Construction of the rankings}\label{app:prefbasis}

The rankings in \cref{tab:dvdstates} were not elicited from the parties;
they are the analyst's reconstruction. They are built in two steps. Each
state is classified by the positions of the two camps (not moved, launched
or accepted the compromise) and by the statement of $T$ (none, endorsement
of MMCD, of SD or of the compromise, or refusal of support). A small number of ordinal
principles, stated below, then order these classes. The two camps are
treated symmetrically: the same principles apply to each, with the roles of
\emph{own} and \emph{rival} exchanged, so that $S$'s rank of a state equals
$M$'s rank of the mirror state. \Cref{tab:dvdcamp} gives the resulting
ranks of the camps.

{\footnotesize
\setlength{\tabcolsep}{4pt}
\begin{longtable}{@{}ll|ccccc@{}}
\caption{Ranks of the camps by situation, from the viewpoint of the own
camp. Each cell gives the rank (1: most preferred) and, in brackets, the
state number when the own camp is $M$ / when it is $S$; --: infeasible.}
\label{tab:dvdcamp}\\
\toprule
& & \multicolumn{5}{c}{Statement of $T$}\\
Own camp & Rival camp & none & endorses own & endorses rival & refusal &
endorses compromise\\
\midrule
not moved & not moved & 6\,{\scriptsize[1/1]} & 6\,{\scriptsize[2/3]} & 6\,{\scriptsize[3/2]} & 6\,{\scriptsize[4/4]} & -- \\
not moved & launched & 8\,{\scriptsize[5/13]} & 5\,{\scriptsize[6/15]} & 12\,{\scriptsize[7/14]} & 6\,{\scriptsize[8/16]} & -- \\
not moved & accepted & 3\,{\scriptsize[9/25]} & 2\,{\scriptsize[10/26]} & -- & 3\,{\scriptsize[11/27]} & 10\,{\scriptsize[12/28]} \\
launched & not moved & 5\,{\scriptsize[13/5]} & 1\,{\scriptsize[14/7]} & 10\,{\scriptsize[15/6]} & 9\,{\scriptsize[16/8]} & -- \\
launched & launched & 10\,{\scriptsize[17/17]} & 5\,{\scriptsize[18/19]} & 12\,{\scriptsize[19/18]} & 11\,{\scriptsize[20/20]} & -- \\
launched & accepted & 5\,{\scriptsize[21/29]} & 1\,{\scriptsize[22/30]} & -- & 9\,{\scriptsize[23/31]} & 10\,{\scriptsize[24/32]} \\
accepted & not moved & 7\,{\scriptsize[25/9]} & -- & 10\,{\scriptsize[26/10]} & 7\,{\scriptsize[27/11]} & 3\,{\scriptsize[28/12]} \\
accepted & launched & 8\,{\scriptsize[29/21]} & -- & 12\,{\scriptsize[30/22]} & 6\,{\scriptsize[31/23]} & 5\,{\scriptsize[32/24]} \\
\multicolumn{2}{@{}l}{unified (both accepted)} & \multicolumn{5}{c}{4\,{\scriptsize[33/33]}} \\
\bottomrule
\end{longtable}
}

\paragraph{Principles for the camps.}
Write $c$ for the own camp and $r$ for the rival.
\begin{enumerate}[label=(C\arabic*),leftmargin=*]
  \item \emph{Support of the TWG.} A format needed the TWG's support for use
        in personal computers \cite{Apple1995}. Hence, for fixed positions of
        the camps, an endorsement of $c$'s position is better for $c$ than
        silence, and an endorsement of $r$'s position worse; a refusal of
        support for $c$'s launch is worse than silence and a refusal for
        $r$'s launch better. In a format war, a refusal harms both launchers and lies
        between silence and the endorsement of the rival. An endorsement of
        the compromise counts as support for the camp that has accepted it.
  \item \emph{Statements without products.} When neither camp has moved,
        a statement of $T$ concerns no product on the market and leaves the
        camps' ranks unchanged (rank $6$ throughout).
  \item \emph{Endorsed sole format.} The best outcome for $c$ is its own
        format on the market with the TWG's endorsement and no rival format
        (rank $1$); the worst is the symmetric outcome for $r$, together with
        a format war in which $T$ endorses $r$ (rank $12$).
  \item \emph{First-mover hope.} An unendorsed launch of $c$'s format,
        without a rival format, is better for $c$ than the status quo
        (rank $5$ against $6$), but one refused support by the TWG is worse (rank $9$). This
        reflects the expectation of each camp that its format could win;
        the SD camp is reported to have regarded its format as the likely
        winner \cite{Wolpin2015}.
  \item \emph{Unification.} A unified format is better for each camp than the
        status quo and than an unendorsed launch, and worse than an endorsed
        sole format and than the rival's unilateral abandonment of its
        format (rank $4$). That both camps finally agreed
        \cite{Nokia1995,Toshiba1995} is consistent with unification ranking
        above the status quo; its position relative to an unendorsed launch
        is a modelling assumption.
  \item \emph{Format war.} A format war without an endorsement is worse for
        each camp than the status quo (rank $10$); with the endorsement of
        $c$'s format it is as good for $c$ as an unendorsed sole launch
        (rank $5$).
  \item \emph{Concession.} Accepting the compromise while the rival holds
        out is worse for $c$ than the status quo (rank $7$) unless $T$
        endorses the compromise. If $T$ does, $c$ is ranked above
        unification (rank $3$) and $r$ near the bottom (rank $10$): the camp
        that accepted first has the TWG's backing in the negotiation on the
        content of the unified format, while the hold-out is isolated.
        Correspondingly, the rival's unilateral concession without such an
        endorsement is good for $c$ (ranks $2$--$3$). This principle is a
        modelling assumption; the sources report the pressure on the
        hold-out camp \cite{Wolpin2015,Doyle1996} but not the camps'
        valuations.
\end{enumerate}

\paragraph{Rankings of the TWG types.}
The rankings of the five candidate types follow from their definitions in
\cref{sec:dvd}; they describe what a camp may believe about $T$, not
evidence about $T$'s actual preferences.
\begin{itemize}[leftmargin=*]
  \item $\theta_B$ (\emph{blocker}): unification is best and the refusal of
        support for a launched format next (rank $2$); then come the states without a
        launch, in which silence and refusal (rank $3$) are preferred to
        endorsements (rank $4$); a launch left unanswered (rank $5$) and an
        endorsed launch (rank $6$) are worst.
  \item $\theta_U$ (\emph{enforcer}): as $\theta_B$, except that endorsing
        the compromise after one camp has accepted it is ranked second only
        to unification (rank $2$).
  \item $\theta_{US}$: as $\theta_U$, except that endorsing SD while no format
        has been launched or accepted is ranked with the refusal of support for a launch
        (rank $3$): the
        enforcer's fallback when the stalemate persists.
  \item $\theta_S$ and $\theta_M$ (\emph{partisans}): endorsing the
        favoured format is best when that format has been launched (rank $1$)
        and next best otherwise (rank $2$); unification is ranked third, silence
        while no format has been launched fourth, and an endorsement of the
        other format last.
\end{itemize}

\paragraph{Which comparisons the results use.}
The stability judgements of \cref{tab:dvdstatic} rest on a few comparisons,
those stated in \cref{sec:dvd}: (i) $\ell_M\succ_M s_0\succ_M b_M,\,e^S_M$
and $e^M_M\succ_M s_0$ (C1, C4); (ii) the best response of each type of $T$
at $\ell_M$: refusal of support for $\theta_B,\theta_U,\theta_{US}$, endorsement of SD
for $\theta_S$ and of MMCD for $\theta_M$; (iii) the rival's responses at
$\ell_M$ are not improvements for it, since a launch in reply leads to a
format war (rank $8$ to $10$ for the rival, C6) and an acceptance leaves its
rank unchanged ($8$), and an acceptance does not lower the launcher's rank
below that of the status quo ($5$ against $6$); (iv) after $e^M_M$ the rival
is at its lowest rank, so it has no improving move; and the symmetric
comparisons for $S$. Any rankings that satisfy these comparisons give the
same judgements; in particular, the remaining choices (the exact position
of unification, the principle (C7) and the ties) do not affect them. The
anticipatory analysis of \cite{KatoAnticipation} uses further comparisons,
notably (C5) and (C7).

\sloppy


\begin{thebibliography}{99}

\bibitem{Aljefri2018}
Y.~M. Aljefri, M.~A. Bashar, L.~Fang, and K.~W. Hipel,
``First-level hypergame for investigating misperception in conflicts,''
\emph{IEEE Trans. Syst. Man Cybern.: Syst.}, vol.~48, no.~12,
pp.~2158--2175, 2018.

\bibitem{Belnap1977}
N.~D. Belnap,
``A useful four-valued logic,''
in \emph{Modern Uses of Multiple-Valued Logic}, J.~M. Dunn and G.~Epstein,
Eds. Dordrecht: Reidel, 1977, pp.~5--37.

\bibitem{Brams1977}
S.~J. Brams,
``Deception in $2\times 2$ games,''
\emph{Conflict Manag. Peace Sci.}, vol.~2, no.~2, pp.~171--203, 1977.

\bibitem{CrawfordSobel1982}
V.~P. Crawford and J.~Sobel,
``Strategic information transmission,''
\emph{Econometrica}, vol.~50, no.~6, pp.~1431--1451, 1982.

\bibitem{Dunn1976}
J.~M. Dunn,
``Intuitive semantics for first-degree entailments and `coupled trees',''
\emph{Philos. Stud.}, vol.~29, pp.~149--168, 1976.

\bibitem{FHK1993}
L.~Fang, K.~W. Hipel, and D.~M. Kilgour,
\emph{Interactive Decision Making: The Graph Model for Conflict Resolution}.
New York: Wiley, 1993.

\bibitem{FraserHipel1979}
N.~M. Fraser and K.~W. Hipel,
``Solving complex conflicts,''
\emph{IEEE Trans. Syst. Man Cybern.}, vol.~9, no.~12, pp.~805--816, 1979.

\bibitem{FraserHipel1984}
N.~M. Fraser and K.~W. Hipel,
\emph{Conflict Analysis: Models and Resolutions}.
New York: North-Holland, 1984.

\bibitem{Ginsberg1988}
M.~L. Ginsberg,
``Multivalued logics: A uniform approach to reasoning in artificial
intelligence,''
\emph{Comput. Intell.}, vol.~4, no.~3, pp.~265--316, 1988.

\bibitem{Howard1971}
N.~Howard,
\emph{Paradoxes of Rationality: Theory of Metagames and Political Behavior}.
Cambridge, MA: MIT Press, 1971.

\bibitem{Inohara2008}
T.~Inohara and K.~W. Hipel,
``Coalition analysis in the graph model for conflict resolution,''
\emph{Syst. Eng.}, vol.~11, no.~4, pp.~343--359, 2008.

\bibitem{Inohara1997}
T.~Inohara, S.~Takahashi, and B.~Nakano,
``Impossibility of deception in a conflict among subjects with
interdependent preference,''
\emph{Appl. Math. Comput.}, vol.~81, no.~2--3, pp.~221--244, 1997.

\bibitem{Kato2022}
Y.~Kato,
``State definition for conflict analysis with four-valued logic,''
in \emph{Proc. IEEE Int. Conf. Syst. Man Cybern.}, 2022, pp.~3186--3191,
doi:10.1109/SMC53654.2022.9945371.

\bibitem{KatoAnticipation}
Y.~Kato,
``Anticipation and disturbance plans in the graph model for conflict
resolution,'' companion manuscript, 2026.

\bibitem{KatoCGMCR}
Y.~Kato,
``Categorical approach to conflict resolution: A corrected correspondence
between category theory and the graph model for conflict resolution,''
arXiv:2306.13961v2, 2026.

\bibitem{KatoExclusion}
Y.~Kato,
``Genuine and spurious equilibria in the graph model for conflict
resolution under effective state exclusion,''
manuscript, version~3.3.

\bibitem{KatoQCW}
Y.~Kato,
``A four-valued graph model for conflict resolution: Core framework and a
machine-checked formalization in Lean~4,'' arXiv:2609.11174, 2026.

\bibitem{Kilgour2001}
D.~M. Kilgour, K.~W. Hipel, X.~Peng, and L.~Fang,
``Coalition analysis in group decision support,''
\emph{Group Decis. Negot.}, vol.~10, no.~2, pp.~159--175, 2001.

\bibitem{KHF1987}
D.~M. Kilgour, K.~W. Hipel, and L.~Fang,
``The graph model for conflicts,''
\emph{Automatica}, vol.~23, no.~1, pp.~41--55, 1987.

\bibitem{Nash1950}
J.~F. Nash,
``Equilibrium points in $n$-person games,''
\emph{Proc. Natl. Acad. Sci. USA}, vol.~36, no.~1, pp.~48--49, 1950.

\bibitem{Nash1951}
J.~Nash,
``Non-cooperative games,''
\emph{Ann. Math.}, vol.~54, no.~2, pp.~286--295, 1951.

\bibitem{Obeidi2009}
A.~Obeidi, D.~M. Kilgour, and K.~W. Hipel,
``Perceptual graph model systems,''
\emph{Group Decis. Negot.}, vol.~18, pp.~261--277, 2009.

\bibitem{RegoVieira2022}
L.~C. Rêgo and G.~I.~A. Vieira,
``Interactive unawareness in the graph model for multilateral conflicts,''
\emph{Discrete Appl. Math.}, vol.~318, pp.~31--46, 2022.

\bibitem{WHF1989}
M.~Wang, K.~W. Hipel, and N.~M. Fraser,
``Solution concepts in hypergames,''
\emph{Appl. Math. Comput.}, vol.~34, no.~3, pp.~147--171, 1989.

\bibitem{Wolpin2015}
S.~Wolpin,
``The history of DVD,'' 2015 (reprint of an article in \emph{DVD Etc.},
2007), \url{https://stewartwolpin.com/the-history-of-dvd/}.

\bibitem{Xu2010}
H.~Xu, D.~M. Kilgour, K.~W. Hipel, and G.~Kemkes,
``Using matrices to link conflict evolution and resolution in a graph
model,'' \emph{Eur. J. Oper. Res.}, vol.~207, no.~1, pp.~318--329, 2010.

\bibitem{Zhang2025}
H.~Zhang, R.~Zhang, J.~Liu, and Y.~Zhu,
``Strategic manipulation behavior in graph model for conflict resolution:
The role of social trust network,''
\emph{Inf. Sci.}, vol.~686, art.~121313, 2025,
doi:10.1016/j.ins.2024.121313.

\bibitem{Apple1995}
Apple Computer,
``Requirements for future high-capacity compact-disc format announced by
computer industry technical group,'' press release, 3 May 1995. Archived at
\url{https://web.archive.org/web/19981202113012/http://product.info.apple.com/pr/press.releases/1995/q3/950503.pr.rel.cd.html}.

\bibitem{Doyle1996}
B.~Doyle,
``A Digital Video Disc standard for the 21st century,'' Desktop Video
Group, undated [c.~1996], \url{https://www.dtvgroup.com/DigVideo/DVD/DVDstory.html}.

\bibitem{LAT1995}
A.~Harmon and D.~Holley,
``Agreement reached on a new format for video: Technology: Toshiba and
Time Warner end battle with Sony and Philips on disc for use in devices to
replace VCRs,'' \emph{Los Angeles Times}, 16 September 1995.

\bibitem{Nokia1995}
Nokia,
``Nokia welcomes single standard for next generation high density optical
disc format,'' press release, 26 September 1995. Archived at
\url{https://web.archive.org/web/19961220013314/http://www.nokia.com/news/news_htmls/nokia_950926.html}.

\bibitem{Taylor1998}
J.~Taylor,
\emph{DVD Demystified: The Guidebook for DVD-Video and DVD-ROM}.
New York: McGraw-Hill, 1998.

\bibitem{Toshiba1995}
Toshiba Corporation,
``DVD format unification,'' press release, 8 December 1995,
\url{https://www.global.toshiba/ww/news/corporate/1995/12/pr0802.html}.

\bibitem{UPI1995}
United Press International,
``Sony, Toshiba agree on DVD format,'' 16 September 1995.

\bibitem{Aljefri2020}
Y.~M. Aljefri, K.~W. Hipel, and L.~Fang,
``General hypergame analysis within the graph model for conflict
resolution,'' \emph{Int. J. Syst. Sci.: Oper. Logist.}, vol.~7, no.~1, 2020,
doi:10.1080/23302674.2018.1476604.

\bibitem{Gharesifard2014}
B.~Gharesifard and J.~Cort\'es,
``Stealthy deception in hypergames under informational asymmetry,''
\emph{IEEE Trans. Syst. Man Cybern.: Syst.}, vol.~44, pp.~785--795, 2014.

\bibitem{Kovach2015}
N.~S. Kovach, A.~S. Gibson, and G.~B. Lamont,
``Hypergame theory: A model for conflict, misperception, and deception,''
\emph{Game Theory}, vol.~2015, art.~570639, 2015.

\end{thebibliography}
\end{document}